\documentclass[a4paper,11pt]{article}  

\usepackage{amssymb,amsthm}
\usepackage{xspace}
\usepackage{comment}
\usepackage{todonotes}

\usepackage{mathtools}
\usepackage{tikz}
\usetikzlibrary{shapes}
\usepackage[numbers,sort]{natbib}
\usepackage{enumerate}
\usepackage{thmtools, thm-restate}
\declaretheorem{definition}
\declaretheorem{Frame Rule}
\declaretheorem{Observation}
\declaretheorem{Fitting Rule}
\newtheorem{lemma}{Lemma}
\newtheorem{theorem}{Theorem}
\newtheorem{proposition}{Proposition}
\usepackage{enumitem}
 \setitemize{itemsep=0.1pt,topsep=4pt,parsep=1pt,partopsep=1pt,leftmargin=10pt,itemindent=10pt}
 \setenumerate{itemsep=0.1pt,topsep=4pt,parsep=1pt,partopsep=1pt}
\usepackage[stretch=10,shrink=10]{microtype}

\usepackage{algorithm,algorithmic}
\usepackage[margin=1in]{geometry}

\newcommand{\lii}[1]{{\small \textit{\phantom{1}#1}} \>}

\newcommand{\com}[1]{// {\it #1}}%{\ensuremath{\triangleright}}
\newcommand{\windowstime}{\ensuremath{(2k)^{4k^2}\cdot k^{O(k)}}}
\newcommand{\overalltime}{\ensuremath{k^{21k^2}\poly(n)}}
\newcommand{\textttup}[1]{\texttt{\upshape #1}}
\newcommand{\splitf}{\textttup{split}\xspace}
\newcommand{\framesf}{\textttup{frames}\xspace}
\newenvironment{noSpaceTabbing}
  {\setlength{\topsep}{0pt}%
   \setlength{\partopsep}{0pt}%
   \tabbing}
  {\endtabbing}

\makeatletter
\renewcommand\bibsection%
{
  \section*{\refname
    \@mkboth{\MakeUppercase{\refname}}{\MakeUppercase{\refname}}}
}
\makeatother

\newcommand{\rep}{\ensuremath{\text{rep}}}
\newcommand{\rel}[1]{\ensuremath{{#1}^*}}
\newcommand{\til}[1]{\ensuremath{#1'}}
\newcommand{\betacrit}{$\beta$-critical\xspace}

\newcommand{\cons}[0]{\ensuremath{\mathcal C}\xspace} 

\newcommand{\leftbreak}[1]{\ensuremath{\operatorname{l_{break}}(#1)}} 
\newcommand{\rightbreak}[1]{\ensuremath{\operatorname{r_{break}}(#1)}} 

\newcommand{\rext}[1]{\ensuremath{\operatorname{r_{ext}}(#1)}} 
\newcommand{\lext}[1]{\ensuremath{\operatorname{l_{ext}}(#1)}} 

\newcommand{\size}[1]{\ensuremath{\overline{#1}}}
\newcommand{\length}[1]{\ensuremath{\left\|{#1}\right\|}}
\newcommand{\shiftr}[2]{\ensuremath{#1\triangleright #2}} 
\newcommand{\shiftl}[2]{\ensuremath{#1\triangleleft #2}} 

\newcommand{\reprep}[0]{\ensuremath{\rep{-}\rep}} 

\newcommand{\sol}[0]{\ensuremath{\mathcal P}\xspace} 

\newcommand{\piecelen}{\ensuremath{\lceil \beta /3\rceil}}

\DeclareMathOperator{\poly}{poly}

\newcommand{\lpiece}[2]{\ensuremath{#1^\triangleleft_{#2}}}
\newcommand{\rpiece}[2]{\ensuremath{#1^\triangleright_{#2}}}

\begin{document}

\graphicspath{{./Figures/}}

\title{\Large Minimum Common String Partition Parameterized by Partition Size
  Is Fixed-Parameter Tractable}
\author{Laurent Bulteau, 
\thanks{Universit\'e de Nantes, LINA, UMR CNRS 6241, Nantes, France. L.B.~was supported by the DAAD.} \\
\and Christian Komusiewicz
\thanks{Institut f\"ur  Softwaretechnik und Theoretische Informatik, TU Berlin, Germany. C.K.~was supported by a post-doc fellowship of the region Pays de la Loire.}}
\date{}

\maketitle

%\pagenumbering{arabic}
%\setcounter{page}{1}%Leave this line commented out.

\begin{abstract}
  The NP-hard \textsc{Minimum Common String Partition} problem  asks whether two strings~$x$ and~$y$ can each be
  partitioned into at most~$k$ substrings
  such that both partitions use exactly the same substrings in a different
  order.  We present the first fixed-parameter
  algorithm for \textsc{Minimum Common String Partition} using only parameter~$k$.
  %}
\end{abstract}
% \begin{keywords}
% %lncs%{\bf Keywords}
% %lncs%{\small
%  NP-hard problem, comparative genomics, parameterized complexity%}
%\end{keywords}
%lncs%\end{minipage}

%lncs%\end{center}

\section{Introduction}
Computing the evolutionary distance between two genomes is a
fundamental problem in comparative genomics~\cite{FLR+09}. Herein, the genomes are usually represented as either strings or permutations and the task is
to determine how many operations of a certain kind are needed to
transform one genome into the other. If the input is a pair of
permutations, these problems can be formulated as sorting problems,
 such as \textsc{Sorting by
  Transpositions}~\cite{BP98} and \textsc{Sorting by
  Reversals}~\cite{BP96}. In this work, we study a problem in this
context whose input is a pair of strings~$x$ and~$y$. Informally, the
operation to transfer~$x$ into~$y$ is to
cut~$x$ into nonoverlapping substrings and to reorder these substrings
such that the concatenation of the reordered substrings is
exactly~$y$. This transformation is formalized by the notion of
\emph{common string partition} (CSP): a partition~$\sol$ of two
strings~$x$ and~$y$ into \emph{blocks}~$x_1 x_2 \cdots x_k$ and~$y_1 y_2 \cdots y_k$
is a common string partition if there is a bijection~$M$
between~$\{x_i\mid 1\le i\le k\}$ and $\{y_i\mid 1\le i\le k\}$ such
that~$x_i$ is the same string as~$M(x_i)$ for all~$1\le i\le k$ (see Figure~\ref{fig:CSP} for an
example). Herein,~$k$ is called the
\emph{size} of the common string partition~$\sol$. We study the problem
of finding a minimum-size CSP:
\begin{quote}
  \textsc{Minimum Common String Partition (MCSP)}\\
  Input: Two strings~$x$ and~$y$ of length~$n$, and an integer~$k$.\\
  Question: Is there a common string partition (CSP)~$\sol$ of size at most~$k$
  of~$x$ and~$y$?
\end{quote}
% \paragraph*{Related Work.}
% \label{sec:related-work}
MCSP was introduced independently by~\citet{CZF+05}
and~\citet{SMEM08} (who call the problem \textsc{Sequence Cover}). MCSP is NP-hard
and
APX-hard even when each letter occurs at most twice~\cite{GKZ05}.  \citet{Dam08}
initiated the
study of MCSP in the context of parameterized algorithmics by showing
that MCSP is fixed-parameter tractable with respect to the combined
parameter ``partition size~$k$ and repetition number~$r$ of the
input strings''. Subsequently, \citet{JZZ+12}  showed that MCSP can be solved
in~$(d!)^k\cdot \poly(n)$ time, where~$d$ is the maximum number
of occurrences of any letter in either input string. This was later improved to a~$d^{2k}\poly(n)$-time algorithm that can solve biologically relevant input instances~\cite{BFKR13}. MCSP can be solved
in~$2^n\cdot \poly(n)$ time~\cite{FJY+11}. The current best approximation ratio is~$O(\log n\log^* n)$ and follows from the  work of~\citet{CM07}. An approximation ratio of~$O(d)$ can also be achieved~\cite{KW07}.

A greedy heuristic for MCSP
was presented by~\citet{SS07}.
In this work, we answer an open question~\cite{Dam08,FJY+11,JZZ+12} by showing that MCSP is fixed-parameter tractable when
parameterized only by~$k$, that is, we present an algorithm with running time~$f(k)\cdot \poly(n)$.
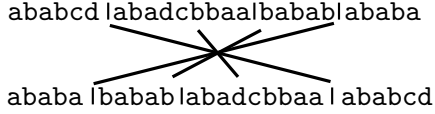
\begin{figure}[t]\centering
%\begin{subfigure}  
  \begin{tikzpicture}[>=stealth,y=16pt,x=32pt]
    \tikzstyle{inv}=[inner sep=0pt,outer sep=0pt,minimum size=0mm]
    \tikzstyle{vert}=[circle,draw,thick, inner sep=0pt,minimum
    size=2mm] \tikzstyle{subtree}=[regular polygon, regular polygon
    sides=3,draw,thick, inner sep=0pt,minimum size=15mm]
    % input strings
    \node (sx1) at (0.9,2) {\texttt{ababcd}}; \node (sx2) at (2.4,2)
    {\texttt{abadcbbaa}}; \node (sx3) at (3.75,2) {\texttt{babab}};
    \node (sx4) at (4.75,2) {\texttt{ababa}};
  
    \node (sy1) at (4.8,0) {\texttt{ababcd}}; \node (sy2) at (3.25,0)
    {\texttt{abadcbbaa}}; \node (sy3) at (1.87,0) {\texttt{babab}};
    \node (sy4) at (0.8,0) {\texttt{ababa}};

    % breakpoints
    \draw [-,thick] (1.54,2.2) -- (1.54,1.8); \draw [-,thick]
    (3.25,2.2) -- (3.25,1.8); \draw [-,thick] (4.24,2.2) --
    (4.24,1.8);

    \draw [-,thick] (1.36,0.2) -- (1.36,-0.2); \draw [-,thick]
    (2.4,0.2) -- (2.4,-0.2); \draw [-,thick] (4.16,0.2) --
    (4.16,-0.2);

    % CSP
    \draw [-,very thick] (sx1) -- (sy1); \draw [-,very thick] (sx2) -- (sy2);
    \draw [-,very thick] (sx3) -- (sy3); \draw [-,very thick] (sx4) -- (sy4);
 
  \end{tikzpicture}\qquad
\caption{\label{fig:CSP}An instance of MCSP with a common string partition of
size four.}
\end{figure}

\paragraph*{Basic Notation.}
A \emph{marker} is an occurrence of a letter at a specific position in
a string; we denote the marker at position~$i$ in a string~$x$
by~$x[i]$. For all~$i$, $1\le i<n$, the markers~$x[i]$ and~$x[i+1]$
are called \emph{consecutive}. An \emph{adjacency} is a pair of
consecutive markers. An \emph{interval} is a set of consecutive
markers, that is, an interval is a set~$\{x[i], x[i+1], \ldots ,
x[j]\}$ for some~$i\le j$. We say that an interval \emph{contains an adjacency} if it contains both markers of the adjacency. We write $[a,b]$ to denote the interval
whose first marker is~$a$ and whose last marker is~$b$. The
\emph{length} $\length{I}$ of an interval $I$ is the number of markers
it contains. Given two markers $a$ and $b$ in the same string~$x$, we
write $\size{ab}$ to denote the signed distance between $a$ and $b$,
that is, $\size{ab}=\length{[a,b]}-1$ if~$a$ appears before~$b$
in~$x$, and~$\size{ab}=-\length{[b,a]}+1$, otherwise. Given two
intervals~$s$ and~$t$, we write $s\equiv t$ if they represent the same
string of letters (if they have the same contents) and~$s=t$ if they
are the same interval, that is, they are substrings of the same string~$z$
and start and end at the same position of~$z$.  Similarly, for
two markers~$a$ and~$b$ we write~$a\equiv b$ if their letters are the
same, and~$a=b$ if the markers are identical. We say that a string~$s$
has \emph{period}~$\pi$ if $s=\rho \pi^i \tau$, where~$i\ge 1$,~$\rho$
is a (possibly empty) suffix of~$\pi$, and~$\tau$ is a (possibly
empty) prefix of~$\pi$. The following lemma is useful when dealing with periodic strings and follows from the periodicity lemma~\cite{FW65}.
\begin{lemma}\label{lem:periodicity}
  Let~$s$ and~$t$ be two strings such that~$s$ has period~$\pi_s$
  and $t$ has period~$\pi_t$. If~$s$ has a suffix of length at least~$\length{\pi_s}+\length{\pi_t}$
  that is also a prefix of~$t$, then~$t$ has period~$\pi_s$ 
  and~$s$ has period~$\pi_t$.
\end{lemma}
We define \emph{offset}
operators~\shiftr{}{} and~\shiftl{}{}: For each marker $e$ and integer~$d$, $e'=\shiftr{e}{d}$ 
is the marker such that $\size{ee'}=d$, and  $\shiftl{e}{d}:=\shiftr{e}{(-d)}$.

\section{Fundamental Definitions and an Outline of the Algorithm.}
In this section, we first present the most fundamental definitions used by our algorithm and then give a brief outline of the main algorithmic strategy followed by the algorithm. 

\paragraph{Some Fundamental Definitions.}
Let~$\sol=\{x_1 x_2 \ldots x_\ell; y_1 y_2 \ldots y_\ell; M\}$ be a
CSP of strings~$x$ and~$y$. A \emph{breakpoint} of~$\sol$ is an
adjacency in~$x$ (or~$y$) that contains the last marker of some block~$x_i$
($y_i$) and the first marker of the next block~$x_{i+1}$
($y_{i+1}$). We say that \sol~\emph{matches two blocks}~$x_i$
and~$y_j$ if~$M(x_i)=y_j$. Furthermore, we say that~\sol~\emph{matches
  two markers}~$a$ and~$b$ if~$a$ and~$b$ are at the same position
in matched blocks. By the definition of a~CSP, this implies~$a\equiv
b$.

The algorithm works on subdivisions of both strings into shorter
parts. These subdivisions are formalized as follows.
\begin{definition}
  A \emph{splitting} of a string (or an interval) $z$ is a list of
  intervals $[a_1,b_1]$, $[a_2,b_2]$, \ldots, $[a_m,b_m]$, each of length at least two, called \emph{pieces} such that $a_1=z[1]$, $a_{j+1}=b_{j}$ for all~$j< m$, and $b_m=z[\length{z}]$.
\end{definition}
Note that successive pieces overlap by one marker. Thus, every adjacency 
of~$z$ is included in one piece. In fact, 
a splitting can be seen as a partition of the adjacencies of a string
(or an interval) such that each part contains only consecutive
adjacencies.

The strategy of the algorithm is to infer more and more information
about a small CSP. To put it another way, it makes more and more
restrictions on the CSP that it tries to construct. To this end, the
algorithm will annotate splittings as follows: a piece
is called \emph{fragile} if it contains at least one breakpoint, and
\emph{solid} if it contains no breakpoint. To simplify the
representation, the algorithm sometimes \emph{merges} consecutive
pieces~$[a_{i},b_{i}]$ and~$[a_{i+1},b_{i+1}]$ (where~$b_i=a_{i+1}$) into one, that is, it
removes~$[a_i,b_i]$ and~$[a_{i+1},b_{i+1}]$ from some splitting and
adds the interval~$[a_i,b_{i+1}]$ to this splitting. 

To further restrict the CSP, the algorithm finds pairs of solid pieces
in~$x$ and~$y$ that are contained in blocks that are matched by the
CSP. Accordingly, a pair of solid pieces~$s$ in~$x$ and~$t$ in~$y$ is
called \emph{matched} in a CSP~$\sol$ if~$s$ is contained
in a block of \sol that is matched to a block that
contains~$t$. Note that 
matched solid pieces may correspond to different parts of their
blocks. For example, one piece may contain the first marker but not
the last marker of its block in~$x$ and it can be matched to a solid
piece that contains the last but not the first marker of its block
in~$y$. Hence, when looking at the two blocks containing the pieces,
%%%changes
there can be a ``shift'' between the matched pieces. We formalize this 
notion by identifying \emph{reference markers}, which are meant to 
be mapped to each other, as follows (see Figure~\ref{fig:alignment} (left) for an example).

\begin{definition}
  Let~$[a,b]$ be a piece of a splitting of~$x$ and~$[c,d]$ be a piece
  of a splitting of~$y$. The \emph{alignment} of~$[a,b]$ and~$[c,d]$
 of \emph{shift}~$\delta$ is the pair of \emph{reference markers}~$a$ and~$\shiftr{c}\delta$, where
\begin{itemize}
 \item $(-\size{ab})\leq \delta \leq \size{cd}$,
 \item $[a,b]\equiv[\shiftr{c}{\delta},\shiftr{c}{(\size{ab}+\delta)}]$, and
 \item 
$[c,d]\equiv[\shiftl{a}{\delta},\shiftl{a}{(\delta-\size{cd})}]$.
\end{itemize}
\end{definition}

Hence, an alignment fixes how the interval~$[a,b]$ is shifted with
respect to~$[c,d]$ in the matched blocks that contain the
intervals. That is, if~$[a,b]$ starts at position~$j$ in its block,
then~$[c,d]$ starts at position~$j-\delta$.  For matched solid pieces,
an alignment thus fixes which markers are matched to each other by the
CSP. In particular, the marker~$a$ is matched to~$\shiftr{c}\delta$ and
$c$ is matched to~$\shiftl{a}\delta$.  Note that the maximum and
minimum values allowed for~$\delta$ ensure that there is at least one
marker in~$[a,b]$ that is matched to a marker in~$[c,d]$ by a CSP
corresponding to this alignment. The algorithm will only consider such
alignments between matched solid pieces. The second condition verifies
that all pairs of matched markers indeed correspond to the same
letter. Clearly, this restriction is fulfilled by every CSP that does
not put breakpoints in the solid pieces~$[a,b]$ and~$[c,d]$.  A pair
of matched solid pieces is called \emph{fixed} if it is associated
with an alignment (equivalently, with a pair of reference markers) and
\emph{repetitive} otherwise (the reason for choosing this term will
be given below).  For a fixed solid piece~$s$, we use~$\rel{s}$ as
shorthand for the uniquely determined reference marker of the
alignment of~$s$ which is in the same string as~$s$.
\begin{figure*}[t]\centering 
 %\begin{subfigure}[t]{0.29\textwidth}
 %\hfill
   \includegraphics{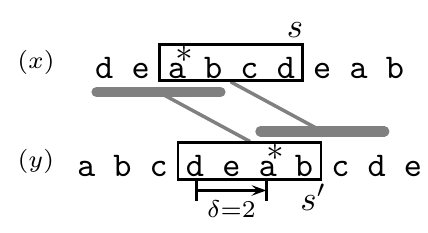}
\hfill
 %\caption{\label{fig:alignment}} 
 %\end{subfigure}
 %\begin{subfigure}[t]{0.7\textwidth}
   \includegraphics{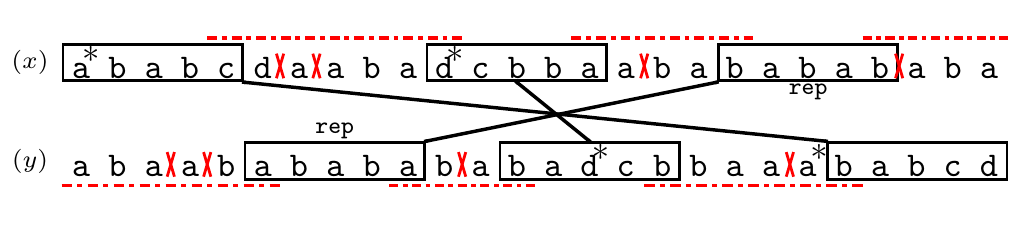}
%\hfill\makebox[0pt]{}

 %  \end{subfigure} 
  \caption{\small 
Left: 
  Example of alignment between two pieces $s$ and $s'$. 
  Reference markers are marked with a star, the shift is 2. 
  Intervals having the same content as the pieces according 
  to this alignment are marked in gray. 
  Note that there also exists an alignment of shift $-3$, 
  where the reference marker in $y$ is the first occurrence 
  of $a$.
Right:~A constraint with three pairs
  of solid pieces illustrated by boxes. Two of these pairs are
  fixed and one is repetitive (\texttt{rep}). 
  Matched solid pieces are linked with edges. 
  The fragile pieces (red and dashed lines) contain
  the breakpoints (red crosses) of a size-5 CSP satisfying the constraint.
}
 \label{fig:example-csp-cons}
 \label{fig:alignment}
\end{figure*}

These restrictions on a possible CSP are summarized in the notion of
constraints, defined as follows, see Figure~\ref{fig:example-csp-cons} 
(right) for an example.
\begin{definition}
  A \emph{constraint} $\cons$ is  a tuple
  $(S,F,M,R_S)$ such that:
  \begin{itemize}
  \item $S$ is a set of solid pieces. Let $S_x$
    ($S_y$) denote the pieces of $S$ from $x$ ($y$).
  \item $F$ is a set of fragile pieces. Let $F_x$
    ($F_y$) denote the pieces of $F$ from $x$ ($y$).
  \item The pieces of $S_x\cup F_x$ ($S_y\cup F_y$) form a
    splitting of $x$ ($y$) in which solid and fragile pieces alternate.
  \item $M:S_x\rightarrow S_y$ is a \emph{matching}, that is, a
    bijection between~$S_x$ and~$S_y$. As shorthand,
    we write~$\til s=M(s)$ if~$s\in S_x$ and~$\til s=M^{-1}(s)$ if~$s\in S_y$.
  \item $R_S$ is a set of alignments that 
    contains for each matched pair of solid pieces at most one alignment.
  \end{itemize}
\end{definition}
Our algorithm will search for CSPs that satisfy such constraints. 
\begin{definition}
  A CSP \sol\ \emph{satisfies} the constraint
  $\cons=(S,F,M,R_S)$ if:
  \begin{enumerate}
  \item All breakpoints of \sol\ are contained in fragile pieces.
    \label{prop:dont-break-solid}
  \item Each fragile piece contains at least one breakpoint from~\sol.
    \label{prop:break-fragile}
  \item Matched solid pieces are contained in matched blocks in~\sol.
  \item If $s$ is a fixed solid piece, then markers~$\rel{s}$
    and~$\rel{\til s}$ are matched in~\sol.
  \item If $s$ is a repetitive solid piece, then $s$, $s'$ and the blocks
    containing them in~\sol{} all have the same shortest periods.
    \label{prop:repetitive-period}				
  \end{enumerate}
\end{definition}
Equivalent formulations of Conditions~\ref{prop:dont-break-solid}
and~\ref{prop:break-fragile} are that (1') all solid pieces are
contained in blocks of~\sol, and (2') different solid pieces in the
same string are in different blocks.  Given a CSP \sol\ that
satisfies a constraint \cons, we call a block \emph{undiscovered by \cons{}} if it does not contain a solid piece
(equivalently, if it is contained in a fragile piece). The other blocks are
 called \emph{discovered by~\cons{}}.

Finally, we introduce the following notion that concerns reference
markers and fixed solid pieces.  
\begin{definition}
  Let $s$ and~$\til s$ be fixed matched solid pieces in~$x$ and~$y$. Two
  markers $a$ in $x$ and $b$ in $y$ are \emph{equidistant from~$s$} if
  $\size{s^* a}=\size{{\til s}^*b}$. Similarly, two intervals~$[a,b]$ in~$x$ and~$[c,d]$ in~$y$ are \emph{equidistant from~$s$} if~$a$ and~$c$ are equidistant from~$s$ and~$b$ and~$d$ are equidistant from~$s$.   
\end{definition}
We will use it to talk about the
``local environment'' of the reference markers in both strings. In
particular, with this notation we can identify (sets of) markers that
are matched to each other if they are both in the same block as
the reference markers.

\paragraph{An Outline of the Algorithm and its Main Method.}
We now give a high-level description of the main idea of the
algorithm; the pseudo-code of the main algorithm loop is shown in
Algorithm~\ref{algo:main-pc}.\footnote{Parts of this algorithm, in
  particular the \splitf procedure follow somewhat the approach of
  Damaschke~\cite{Dam08}.} For the discussion, assume that the
instance is a yes-instance, that is, there exists a~CSP~$\sol$ of
size~$k$. Since we can check in polynomial time the size and
correctness of any CSP before outputting it, we can safely assume that
the algorithm gives no output for no-instances; hence the focus on
yes-instances. The algorithm gradually extends a constraint that is
satisfied by a solution~$\sol$ and outputs~$\sol$
eventually. Initially, the constraint consists solely of two fragile
pieces, one containing all of~$x$ and one all of~$y$. We assume that
the input strings are not identical. Hence, every CSP has at least one
breakpoint and the initial constraint is thus satisfied by every
size-$k$ CSP.

The algorithm now aims at discovering the blocks of~$\sol$
successively, from the longest to the shortest. Recall that a block is
called discovered by a constraint \cons if there is a solid piece in
\cons that is contained in this block. To execute the strategy of
finding shorter and shorter blocks, the algorithm needs some knowledge
about the approximate (by a factor of 2) length of the longest
undiscovered block in~$\sol$. To this end, the algorithm keeps and
updates an integer variable~$\beta$ which has the following central
property: Whenever there is a size-$k$ CSP satisfying the current
constraint, then there is in particular one size-$k$ CSP~$\sol$ such
that
\begin{enumerate} 
\item the longest undiscovered block of~$\sol$ has length~$\ell$
with~$\beta\le \ell< 2\beta$, and
\item $\beta$ is minimum among all powers of two satisfying Property~1.
\end{enumerate}
Accordingly, we call a block~\emph{\betacrit }if it has length~$\ell$
with~$\beta\le \ell< 2\beta$. To obtain~$\beta$, we consider all
subsets~$\Pi'$ of the set~$\Pi$ containing all powers of 2 that are
smaller than~$n$. One of these sets will contain the ``correct''
approximate block lengths. The central strategy is: Set~$\beta$ to be
the largest value in~$\Pi'$. Discover all \betacrit blocks. Then,
there is a satisfying CSP such that all undiscovered blocks are
shorter than the current~$\beta$. Thus update~$\beta$ by taking the
next largest value from~$\Pi'$. Then, again discover all~\betacrit
blocks, update~$\beta$ again and so on. The crucial algorithmic trick
is to use the fact that for the correct~$\Pi'$ we know that there is
at least one \betacrit block and that this block is the largest of all
undiscovered blocks. Hence, this is why we guess~$\Pi'$ instead of
considering the set~$\Pi$ of all powers of two in descending order.
%laurent%
Note that since the values of~$\beta$ are decreasing in the course of
the algorithm, branching into all possible $\Pi'\subseteq \Pi$ in the beginning is
equivalent to performing a branching at the point where the next
$\beta$ is ``needed''. By fixing~$\Pi'$ in advance it is clearer,
however, that the number of branches is $2^{\log n}$ instead of ${\log
  n}^k$.

First, note that there is at least one block of length at
least~$\lceil n/k\rceil$ since~$\sol$ has size~$k$, so $\max \Pi'\geq
\lceil n/2k\rceil$.  Furthermore, for any CSP of size~$k$,~$|\Pi'|\le
k$. Hence, the outer algorithm loop of Algorithm~\ref{algo:main-pc} is
traversed once for the correct~$\Pi'$.
%(we will need this restriction to prove the running time bound). 
Note furthermore, that the number of subsets of~$\Pi$ is~$O(2^{\log n}) =
O(n)$. Hence, there are~$O(n)$ traversals of the outer loop of the
main method.

\begin{algorithm*}[t]
  \caption{The main algorithm loop \textttup{MCSP}$(x,y,k)$.}
  \label{algo:main-pc}
  \begin{noSpaceTabbing}
    \hspace*{0.7cm}\=\hspace*{0.5cm}\=\hspace*{0.5cm}\=\hspace*{0.6cm}\=\hspace*{5.5cm}\=\kill
 %   \textttup{MCSP}$(x,y,k)$\\
%    \textbf{Initialization:}\\ %(or ``\lii{} initialization''?)
    \lii{1} $\Pi:=\{i\in \mathbb{N}\mid i < n \wedge \exists  j\in \mathbb{N}: 2^j=i\}$\\
\lii{2} $\cons:= \{S:=\emptyset, F:=\{[x[1],x[n]],[y[1],y[n]]\} , M:=\emptyset, R_S:= \emptyset\}$ % \> \>\> 
\quad \com{initially: only two fragile pieces}\\    
    \lii{3} \textbf{for each} $\Pi'\subseteq \Pi$ with $\max \Pi'\ge \lceil n/2k \rceil \wedge |\Pi'|\le k$  \textbf{:}\\
    \lii{4} \> $\beta \leftarrow \max \Pi'$; $\Pi'\leftarrow \Pi'\cup \{0\}\setminus \{\beta\}$ \> \>\> \com{2-approx. length of longest undiscovered block}\\ 
    \lii{5} \> \textbf{repeat until} $\beta < 4$ \textbf{:}\\
    \lii{6} \> \> \splitf  \> \>  \com{discover blocks of length at least~$\beta$}\\
    \lii{7} \> \> $\beta \leftarrow \max \Pi'$; $\Pi'\leftarrow \Pi'\setminus \{\beta\}$  \> \>  \com{update 2-approx. length of longest undiscovered blocks}\\
    % \lii{5}  \> $\beta \leftarrow$ \textttup{guess}$(1,2,4, \ldots , 2^{\lfloor
    %   \log(n)\rfloor})$ \> \> \> \com{approx.~length of longest short block}\\
    % \lii{6}  \> \textttup{new-align} := True\\
    % \lii{7}  \>  \textbf{repeat until} \textttup{new-align} = False \textbf{:}\\
    \lii{8}  \> \> \framesf \> \>  \com{reduce length of fragile pieces} \\
    % \lii{9}  \> \>   $(\cons, \text{\textttup{new-align}}) \leftarrow \text{\textttup{align-long-periods}}(\cons)$\\
    \lii{9} \> branch into all cases to set breakpoints within fragile pieces\\ 
    \lii{10} \> \textbf{if} the resulting string partition~$\sol$ is a size-$k$ CSP \textbf{:}  \textbf{output} $\sol$
  \end{noSpaceTabbing}
  % \caption{}\label{fig:main-pc}
\end{algorithm*}

Consider now the traversal for the correct set~$\Pi'$. The inner loop
of the algorithm consists of two main steps. In the first step, called
\splitf, the algorithm discovers the~\betacrit blocks. More precisely,
it refines~\cons{} by breaking fragile pieces into shorter pieces (of length~\piecelen)
and identifying those that are contained in~\betacrit blocks. 
It then produces a matching and, if this is possible without considering too many options, aligns these blocks.

To be efficient, \splitf{} requires that the input fragile pieces are
short enough compared to~$\beta$ and~$k$. Initially, this is not a problem, since the fragile pieces have length~$n$,
and~$\beta\geq n/2k$.  After \splitf, however, we update~$\beta$. Hence, between two calls to \splitf{}
the fragile pieces have to be reduced in order to fit the undiscovered
blocks more ``tightly''.  This is the objective of \framesf, which
uses a set of rules to identify smaller intervals containing all
breakpoints of~\sol. It thus shrinks the fragile pieces of~\cons{} so
that they are sufficiently small for the next call to \splitf.

The algorithm now continues with this process for smaller and smaller
values of~$\beta$. It stops in case~$\beta<4$, since it can then
locate all breakpoints by applying a brute-force branching. Note that
in order to ensure that there is always a $\beta<4$, we add the
value~0 to the set~$\Pi'$ in Line~4 of the main method.

In the remainder of this work, we give the details for the
procedures~\splitf and \framesf. In Section~\ref{sec:split}, we
%formally define constraints, 
describe the~\splitf procedure, and show its correctness. We also
show, using several properties of~\framesf as a black box, our main
result. Then, in Sections~\ref{sec:frames-aligned}
and~\ref{sec:frames-repetitive}, we fill in the blanks by proving the
properties of~\framesf. 

The algorithm is a branching algorithm that extends the
constraint~$\cons$ in each branch. In order to simplify the
pseudo-code somewhat, we describe the algorithm in such a way that the
variables~$\cons$ and $\beta$ are global variables. 
After a branching statement in the pseudo-code, the algorithm
continues in each branch with the following line of the
pseudo-code. If a  branch is known to be unsuccessful, then
the algorithm returns immediately to the branching statement that
created this branch (or to the branching statement above, if the current branch is the last branch of that statement). We denote this by the ``abort branch'' command; 
all modifications within this branch are
undone. 

% Due to the lack of space, almost all proofs are deferred to an appendix.\footnote{A full version is also available at \url{http://fpt.akt.tu-berlin.de/publications/mcsp.pdf}.} 
\section{Splitting of Fragile Pieces}
\label{sec:split}
In this section, we describe the procedure \splitf and show its
correctness. The pseudo-code of \splitf is shown in
Algorithm~\ref{algo:split-pc}. At the beginning of~\splitf the
constraint contains a set of discovered blocks. Assume that all blocks
of length at least~$2\beta$ are discovered by this constraint. The aim
of~\splitf now is to perform a branching into several cases such that
in at least one of the created branches the
constraint~$\cons$ now additionally contains all~\betacrit
blocks. Hence, in this branch all blocks of length at least~$\beta$
are discovered.
\begin{algorithm*}[t]
  \caption{Procedure \splitf. Global variables: $\cons=(S,F,M,R_S)$ and~$\beta$.}
  \label{algo:split-pc}
  \begin{noSpaceTabbing}
    \hspace*{0.7cm}\=\hspace*{0.5cm}\=\hspace*{0.5cm}\=\hspace*{0.6cm}\=\hspace*{5.5cm}\=\kill
 %   \splitf\\
 %   \lii{1} split each fragile piece from~$\cons$ into pieces of size at most~$\beta/2$ such that: \\
 %   \> \> -- new consecutive pieces overlap by their endpoints, and\\
 %   \> \> -- for each \\
    \lii{1} $N:= \emptyset $\> \> \> \>  \com{the set of new pieces}\\   
    \lii{2} \textbf{for each} fragile piece $f\in F$ \textbf{:}\\
    \lii{3}  \> $F\leftarrow F\setminus \{f\}$  \>  \>  \> \com{old  fragile pieces are removed}\\
    \lii{4}  \> $N\leftarrow N\; \cup $ ``$\lceil  \beta /3  \rceil$-splitting of $f$''  \> \> \>  \com{update set of new pieces} \\
    \lii{5} \textbf{for each} $p\in N$~\textbf{:} \>\>\>\> \com{make $p$ either fragile or solid}\\  
    \lii{6}  \> \> branch into the case that either~$S\leftarrow S\cup \{p\}$ or $F\leftarrow F\cup \{p\}$ \\
    % \> \> \> \com~$P$ is solid or fragile\\
    \lii{7}  \textbf{while} $\exists$ consecutive pieces~$s_1, s_2$ s.t.~$
\{s_1, s_2\}\subseteq S$ (or~$\{s_1, s_2\}\subseteq F$) \textbf{:}\\
    \lii{8} \> $p:=$ ``merged interval of~$s_1$ and~$s_2$''\\
    \lii{9}  \> $S\leftarrow (S\cup p) \setminus \{s_1, s_2\}$ (or $F\leftarrow (F\cup p) \setminus \{s_1, s_2\}$)\\
   %  \lii{10}  \> \textbf{if}~$\{s,t\}\subseteq N$\textbf{:}~$N\leftarrow (N\cup p)\setminus \{s,t\}$ \> \> \> \com{merging new pieces gives a new piece} \\ 
   % \lii{11}  \> \textbf{else :} $N\leftarrow N\setminus \{s,t\}$ \com{merge with an old piece gives only extension of an old piece}\\
   % \lii{12}  \textbf{if} a solid piece~$s$ has length $< \lceil  \beta /2 \rceil$ \textbf{:} \> \> \>\> \com{$s$ contains last marker of~$x$ or~$y$} \\
   %  \lii{13} \> $S\leftarrow S\setminus \{s\}$;~$F\leftarrow F\cup \{s\}$; merge~$s$ with neighboring fragile piece\\
   %  \lii{14}  \textbf{if} $|N\cap S|=0$   \textbf{: return} ``no'' \> \> \> \> \com{no new solid piece found}  \\
    \lii{10}  \textbf{if} $|S_x|\neq |S_y|$ \textbf{: abort branch}  \> \> \> \> \com{no bijection of solid pieces exists}  \\
    \lii{11}  \textbf{if} $|F_x|\ge k$ or $|F_y|\ge k$ \textbf{: abort branch} \> \> \> \> \com{too many fragile pieces in~$x$ (or~$y$)}  \\    
    \lii{12}  \textbf{while} $\exists$ unmatched solid piece~$s\in S_x$ \textbf{:}  \\
    \lii{13} \>  \textbf{for each} unmatched solid piece~$t$ in~$S_y$ \textbf{:}\\ 
    \lii{14} \> \> branch into the case that~$M(s):=t$\\
    \lii{15}  \textbf{for each} new pair~$(s,t)$ of matched solid pieces \textbf{:}\\
    \lii{16} \> $i:=$ ``number of alignments with shift $\delta$ s.t.~$|\delta|\leq \piecelen$''\\
    %\lii{17} \> \textbf{if} $i=0$ \textbf{: return} ``no''\\
    \lii{17} \> \textbf{if} $i\le 6$ \textbf{:}  for each alignment~branch into the case to add this alignment to~$R_S$\\
    \lii{18} \> \textbf{else:} branch into the cases to: \> \>  \> \com{$s$ and~$\til s$ are periodic}\\
    \> \> \> -- align~$s$ and~$\til s$ such that~$\leftbreak{s}$ and~$\leftbreak{\til s}$ are equidistant from~$s$ \\ 
    \> \> \> -- align~$s$ and~$\til s$ such that~$\rightbreak{s}$ and~$\rightbreak{\til s}$ are equidistant from~$s$\\
    \> \> \> -- do not align~$s$ and~$\til s$
\end{noSpaceTabbing}
\end{algorithm*}
Procedure \splitf{} starts by replacing each former fragile piece 
by a splitting where all new pieces have length~\piecelen{} except for the
 rightmost new piece of each such splitting which can be shorter. We call such a
splitting a \emph{$\piecelen$-splitting}. It then considers all branches
where each piece is either fragile or solid. In order to maintain the 
alternating condition, consecutive solid (resp. fragile) pieces are
merged into one solid (fragile) piece, Lines~7--9.

Next, \splitf extends the matching and the set of
alignments of the constraint.  All possible matchings are considered in separate
branches (Lines~12--14).  Then, \splitf{} performs an exhaustive
branching over all alignments for a given pair of solid pieces, but
only if there are very few of them (Line~17). If there are too many
(Line~18), then it can be seen that the pieces are periodic with a
short period length. Thus, the blocks containing them might be periodic as
well.  If the blocks are not periodic, then there are at most two
alignments that the algorithm needs to consider: informally, the
period in the blocks can be ``broken'' either to the left or to the
right of the pieces.  To specify these two possibilities more clearly,
we introduce the following notation. Let~$s=[a,b]$ be an interval in a
string~$x$ such that~$s$ has a shortest period~$\pi$. Then, we denote
by~$\leftbreak{s}$ the rightmost marker in~$x$ such
that~$[\leftbreak{s},b]$ does not have period~$\pi$. Similarly, let
$\rightbreak{s}$ be the leftmost marker in~$x$ such
that~$[a,\rightbreak{s}]$ does not have period~$\pi$.  If the blocks
are periodic, there may be too many possible alignments, and the
alignment between the pieces will be fixed at a later point (when
$\beta$ becomes smaller than the period). However, the algorithm will
use the ``knowledge'' that the blocks are periodic in the \framesf
procedure.

We now show that \splitf is correct if the input constraint can be
satisfied and that it discovers all \betacrit{} blocks. 
\begin{lemma} \label{lem:split-correct} Let~$\cons$ be the constraint
  at the beginning of~\splitf, and let~\sol be a size-$k$ CSP satisfying \cons
  such that all blocks of length at least~$2\beta$ of \sol are
  discovered by~$\cons$. Then, \splitf creates at least one 
  branch whose constraint $\cons$
  \begin{itemize}
  \item is satisfied by $\sol$, and
  \item all blocks of length at least~$\beta$ are discovered by $\cons$
  \end{itemize}
  % Moreover, \splitf is called $\leq k$ times.
\end{lemma}
\begin{proof} 
  Let~$B = \{(x^1,y^1), \ldots ,(x^\ell,y^\ell)\}$ be the uniquely
  defined set of matched pairs of undiscovered blocks in~$\sol$ that
  are~\betacrit.

  Consider the following branching for Lines~5--6 for each piece~$p\in N$: If~$p$ is  contained in some block~$x^i$ or~$y^i$
  of~$B$, then branch into the case that~$p$ is added
  to~$S$. Otherwise, branch into the case that~$p$ is added to~$F$ 
	(note that we may add in $F$ some pieces that do not contain any breakpoint, 
	but are contained in blocks not in $B$).
  
  Now consider the constraint obtained for the above branching after
  the merging operations performed in Lines~7--9. We show that~$\sol$
  satisfies Conditions~1 and~2 of this constraint. First, consider a
  breakpoint in~$\sol$. This breakpoint is contained in some fragile
  piece~$f$ of the input constraint since~$\sol$ satisfies this input
  constraint. Hence, it is contained in some new piece~$p$ of the
  splitting of this fragile piece. Clearly, the piece~$p$ is added
  to~$F$ in the considered branching. Moreover, in case~Lines~7--9
  merge fragile pieces, the resulting piece is also fragile, hence~$p$
  remains in a fragile piece. Consequently, all breakpoints of~$\sol$
  are in fragile pieces of~$F$, and thus Condition~1 is satisfied by~$\sol$.

  Now consider a fragile piece~$f\in F$ after Lines~7--9 of the
  algorithm. We show that~$f$ contains at least one breakpoint. Note
  that~$f$ is obtained after a (possibly empty) series of merging
  operations. After the merging,~$f$ is between two solid
  pieces. If~$f$ is also a fragile piece in the input constraint,
  then~$f$ contains a breakpoint since~$\sol$ satisfies the input
  constraint. Otherwise,~$f$ is contained in a fragile piece of the
  input constraint, and at least one of its neighbor pieces is a new
  solid piece~$s$. Since~$f$ (or all the smaller pieces that were merged
  to~$f$) are added to~$F$ by the branching, they are not contained in
  the block that contains~$s$. Hence,~$f$ contains the
  breakpoint between the first (or last) marker of the block
  containing the new solid piece and its predecessor (or
  successor). Thus,~Condition~2 is also satisfied by~$\sol$.

  Note that the above also implies that, for each~$x^i$ of~$B$, there
  is exactly one new solid piece that is contained
  in~$x^i$. Similarly, for each~$y^i$ of~$B$, there is exactly one new
  solid piece that is contained in~$y^i$. Note that in this branching,
  $|S_x|=|S_y|$ and furthermore, since~\sol has size~$k$, $|F_x|< k$
  and~$|F_y|< k$. Hence, the algorithm does not abort in Lines~10
  and~11. We now consider the branching in which for each
  pair~$(x^i,y^i)$, the two corresponding solid pieces are matched to
  each other. Clearly, this branching fulfills Condition~3: the
  condition holds obviously for all pieces contained in blocks
  of~$B$. Furthermore, it holds for all old solid pieces since for
  these, the matching~$M$ has not changed. Note that the function~$M$
  also remains a bijection: it is changed only for unmatched solid
  pieces, and the number of new solid pieces in~$x$ and~$y$ is equal.

  It remains to show that there is a branching in which Conditions~4
 and~5 also hold. 
  Consider a pair of matched solid pieces~$s$ and $s'$,
 and the blocks $x^i$, $y^i$ containing them. We use the following 
 technical claim in order to clarify the discussion; it will be proven afterwards.
 \begin{quote}
   \emph{Fact.} If there are more than six alignments of~$s$ and~$\til s$ whose
   shift each has absolute value of at most $\piecelen$, then
   \begin{itemize}
   \item[i.] $s$ and~$\til s$ are periodic with a common shortest
     period~$\pi$ (with $\length{\pi}\leq\piecelen/2$);
   \item[ii.] if the blocks $x^i$ and $y^i$ do not have period $\pi$,
     then in \sol\ either~$\leftbreak{s}$ is matched
     to~$\leftbreak{\til s}$, or~$\rightbreak{s}$ is matched
     to~$\rightbreak{\til s}$ (or both).
   \end{itemize}
 \end{quote}
  Let $a$ be the leftmost marker of $s$ and $\tilde a$ be the marker matched to~$a$ 
  in~\sol{}. Then~$\{a,\tilde a\}$ is an alignment for~$(s,s')$
  whose shift has absolute value less than~$\piecelen$:
  there are at most~$\piecelen-1$ markers preceding
  either~$s$ or~$s'$ that can belong to the same
  block since the pieces of the $\piecelen$-splitting
  preceding~$s$ and $s'$ are fragile and thus not contained in the
  same blocks. 	
  If the condition of Line~17 is satisfied, then there is one branch
  where alignment~$\{a,a'\}$ is added to~$R_S$. Otherwise, by 
  the fact above, the following cases 
  are possible. Either~$\{a,a'\}$ is one of the alignments where
 \leftbreak{s} is matched to \leftbreak{s'} or \rightbreak{s}
 is matched to \rightbreak{s'}, in which cases~$\{a,a'\}$ is added to 
 $R_S$ in one of the branches. Otherwise,~$(s,s')$ is not fixed, and~$s$ 
 and~$s'$ are contained in blocks having the same shortest periods.

 Altogether this shows the first claim of the lemma.
 The second claim can be seen as follows. The blocks of
 length~$\ell\ge 2\beta$ are already discovered, and the corresponding
 solid pieces remain in the constraint. It thus remains to consider
 the~\betacrit blocks. We show that for each~$x^i$ there is at least
 one piece that is contained in~$x^i$. Consider the marker~$a$ at
 position~$\piecelen$ in~$x^i$ and a piece~$s$ of the
 \piecelen-splitting that contains this marker.  Then~$s$ contains
 only markers from~$x^i$ since~$s$ has length at most~$\piecelen$
 and~$x^i$ has length at least~$\beta\geq2\piecelen$ (for $\beta\geq
 4$).  Afterwards, $s$ is only merged with other pieces that are
 contained in~$x^i$ (recall that in the considered branching there is
 a fragile piece between all solid pieces from different
 blocks). Hence, the second claim of the lemma also holds.  

It remains to show the correctness of the claimed fact. We first need to prove the following claim. 
Define  the \emph{$\piecelen$-middle} of an interval~$[u,v]$ as 
the length-$\piecelen$ interval centered in $[u,v]$ (formally, the 
interval $[\hat u,\hat v]$ with 
$\hat u=\shiftr u{\lfloor(\size{uv}-\piecelen)/2\rfloor}$ and 
$\hat v= \shiftl v{\lceil(\size{uv}-\piecelen)/2\rceil}$).
Then $s$ contains the $\piecelen$-middle of $x^i$
and $s'$ contains the $\piecelen$-middle of $y^i$.

The claim is shown for $s=[a,b]$ (the proof for $s'$ is similar).
Let $x^i=[u,v]$, and let~$[\hat u,\hat v]$ be the $\piecelen$-middle of $x^i$.
First note that since~$x^i$ has length at least~$\beta$, 
we have $\size{uv}\geq \beta-1$.
We show that
$a$ is in the interval $[ u,\hat u]$. The length of this interval is 
\begin{align*}
\size{u\hat u}
 &=\lfloor(\size{uv}-\piecelen)/2\rfloor\\
 &\geq \lfloor(\beta-1-\piecelen)/2\rfloor\\
 &\geq \lfloor(\lfloor2\beta/3\rfloor-1)/2\rfloor\\
 &\geq \lfloor(2\beta/3-1.7)/2\rfloor\\
 &\geq \lfloor\beta/3-0.85\rfloor\\
 &\geq  \piecelen-2. 
\end{align*}
Since the piece with right endpoint $a$ in the $\piecelen$-splitting is fragile 
(it has not been merged with $s$), 
it contains a breakpoint of \sol{} and hence a marker strictly to the left of $u$. 
Moreover it has length at most $\piecelen$, 
so $\size{ua}\leq \piecelen-2$, which implies that $a$ is in the interval $[u,\hat u]$. 
Similarly, 
$b$ is in the interval $[\hat v,v]$, 
and $[a,b]$ contains the \piecelen-middle of $x^i$.

We can now turn to proving the two statements of the fact.

(i)  Let $s=[a,b]$, $s'=[a',b']$ and  
$\delta_1,\delta_2,\ldots,\delta_m$ be
the shifts of the $m\geq 7$ alignments
such that
$-\piecelen\leq \delta_1\leq \delta_2\leq\ldots\leq \delta_m\leq \piecelen$.
Let $i$ be the index such that $\delta_{i+1}-\delta_i$ is minimal, 
and $p=\delta_{i+1}-\delta_i$. We thus have 
\begin{align*}
p\leq\frac{ 2\piecelen }{m-1} \leq \piecelen/2
\end{align*}
%Recall also that both $s$ and $s'$ have length at least $2\piecelen-1$.
Let $q$ be an integer with $p\leq q \leq \size{ab}$. 
Using the second condition in the definition
of alignment, we have 
\begin{align*}
 \shiftr aq &\equiv \shiftr {a'}{(\delta_i+q)}\quad\mbox{(since $\shiftr aq\in[a,b]$)}\\
      &= \shiftr {a'}{(\delta_{i+1}+q-p)}\\
      &\equiv \shiftr a{(q-p)}\quad\mbox{(since $\shiftr a{(q-p)}\in[a,b]$)}
\end{align*}
Thus intervals $[a,b]$ and (symmetrically) $[a',b']$ 
are both periodic with period length $p$, hence
$s$ and $s'$ have shortest periods of length at most $\piecelen/2$.

Using the fact that $s$ and $s'$ both contain the
\piecelen-middle of the matched blocks in which they are contained,
they have a common substring of length greater than the sum of their 
%laurent%
shortest periods. By~Lemma~\ref{lem:periodicity} they thus have a common shortest period, written 
$\pi$, with $\length{\pi}\leq\piecelen/2$.

(ii)
Recall that $x^i$ ($y^i$) is the block containing $s$ ($s'$) in \sol.
Let $[\hat u,\hat v]$ ($[\hat u',\hat v']$) be the \piecelen-middle of  $x^i$ ($y^i$).
Since $[\hat u,\hat v]\subset s$ and $\length{[\hat u,\hat v]}\geq \length{\pi}$, 
we have that $\leftbreak{s}$ is the rightmost marker in~$x$ and
 $\rightbreak{s}$ is the leftmost marker in~$x$ such
that intervals~$[\leftbreak{s},\hat v]$ and $[\hat u,\rightbreak{s}]$ do not have period~$\pi$.
The same holds for $\leftbreak{\til s}$ ($\rightbreak{\til s}$) and $\hat v'$ ($\hat u'$).

Now if $x^i$ does not have period $x_i$ it either contains \leftbreak{s} or \rightbreak{s}.
Suppose that $x^i$ contains  \leftbreak{s} (the case where
$x^i$ contains \rightbreak{s} is similar).
Let~$l'$ be the marker in $y^i$ matched to \leftbreak{s} by \sol.
Then~$[l',\hat v']\equiv [\leftbreak{s},\hat v]$. Hence,~$[l',\hat v']$ does not have period $\pi$. Furthermore, by the definition of~$\leftbreak{s}$, 
 for all $m'\in[\shiftr{l'}1,\hat v']$, $[m',\hat v']$ has period $\pi$.
Thus, $l'$ is the rightmost marker such that $[l',\hat v']$ does not have period $\pi$. Since $[\hat u',\hat v']\subset \til s$ and since $\til s$ has the same shortest period~$\pi$ as $[\hat u',\hat v']$ we have $l'=\leftbreak{\til s}$.
Hence, markers \leftbreak{s} and \leftbreak{\til s} are matched in \sol. 
\end{proof}

The following trivial observation follows from the check in~Line~11 of
\splitf. It is useful for bounding the running time of \splitf (in
particular for later calls to \splitf).
\begin{Observation}\label{obs:piece-number}
  After \splitf has finished, the constraint contains at most~$2k-2$ fragile pieces from each of~$x$ and~$y$. The overall number of solid pieces is thus at most~$2k$.
\end{Observation}
To obtain a fixed-parameter algorithm for parameter~$k$, we now
``shrink'' the fragile pieces between the solid pieces of the
constraint. This will ensure that in the next call to
\splitf, the number of new pieces created in the splitting will be 
bounded by a function of~$k$. 
Note that by
Lemma~\ref{lem:split-correct},~\splitf has discovered \emph{all}
pieces that have length at least~$\beta$. Hence, we now update the
value~$\beta$ denoting the approximate length of the longest undiscovered
blocks (by taking the largest remaining value from~$\Pi$).  Then,
\framesf uses this updated value of~$\beta$ to shrink the fragile
pieces. For the moment, we make some claims about \framesf; their
proof is deferred to the Sections~\ref{sec:frames-aligned}
and~\ref{sec:frames-repetitive}. First, we claim that~\framesf is correct, that is, there is at least one good branching for yes-instances.
\begin{restatable}{lemma}{windowslem} \label{lem:windows-correct} If
  there exists a size-$k$ CSP \sol\ satisfying \cons\ at the beginning
  of~\framesf such that longest undiscovered block is~\betacrit, then
  \framesf creates at least one branch such that the
  constraint in this branch is satisfied by a size-$k$
  CSP~$\sol'$ whose longest undiscovered block has length at most
  $2\beta-1$.
\end{restatable}
Second,~\framesf increases the exponential part
of the running time by a factor that depends only on~$k$.

\begin{restatable}{lemma}{windowsrt} \label{lem:windows-rt} Overall, the calls to \framesf create~$\windowstime$ branches;
  all other parts of the algorithm can be performed
  in~$\poly(n)$ time.
\end{restatable}

Finally, to bound the number of branches in the subsequent call
to \splitf, and for the case~$\beta<4$, we use the following
lemma.
\begin{restatable}{lemma}{lemframesize}
 \label{lem:frame-size}
 When \framesf terminates, every fragile piece has
 length at most $12(k^2+k)k\beta $.
\end{restatable}
Note that the above also holds before the first call
of \splitf. Using these lemmas, we obtain our main result.
\begin{theorem}
  \textsc{Minimum Common String Partition} can be solved in
  \overalltime~time; it is thus fixed-parameter tractable with respect
  to the partition size~$k$.
  \label{thm:mcsp-fpt}
\end{theorem}
\begin{proof}%[{of Theorem~\ref{thm:mcsp-fpt}}]
  For the correctness proof assume that the instance is a yes-instance
  (for a no-instance the algorithm can always check the correctness
  and size of a CSP before returning, thus it has empty output for no-instances). Then, assuming that the input strings are not
  identical, there is a CSP~\sol  satisfying the initial
  constraint~\cons which demands only that there is at least one
  breakpoint in~$x$ and in~$y$.

  We now show that there is a set~$\Pi'$ of powers of 2, all of which
  are smaller than~$n$ such that the algorithm outputs, in at least
  one of its branches, a size-$k$ CSP, in case the main
  algorithm loop is traversed for this set~$\Pi'$.

  Let~$\beta$ be the smallest integer such that there is a size-$k$
  CSP in which the longest block is \betacrit. Then, the largest
  integer of~$\Pi'$ is~$\beta$. Now, if~$\beta<4$ the algorithm
  directly finds all breakpoints by a brute-force
  branching. Otherwise, the procedure \splitf is called. By
  Lemma~\ref{lem:split-correct}, this procedure creates at least one
  branch where the constraint is satisfied by some
  size-$k$ CSP~$\sol$ and all its blocks of length at least~$\beta$
  are discovered by~$\cons$. Consider an arbitrary branch
  with this property. Now, let~$\beta$ denote the smallest power of 2
  such that there is a CSP satisfying the current constraint~$\cons$
  in which the longest blocks are \betacrit. This integer~$\beta$ is
  the second largest integer of~$\Pi'$. The algorithm now calls
  \framesf and by Lemma~\ref{lem:windows-correct} obtains in at least
  one branch a constraint such that there is a size-$k$ CSP that
  satisfies the constraint in this branch. Furthermore, also by
  Lemma~\ref{lem:windows-correct} the longest undiscovered block in
  this CSP has length at most~$2\beta-1$. By the choice of~$\beta$, it
  follows that the longest undiscovered block of this CSP
  is~\betacrit. Now, the algorithm either finds all breakpoints by
  brute-force (if~$\beta<4$) or again calls the procedure \splitf to
  discover all~\betacrit blocks. This whole procedure is repeated for
  smaller and smaller~$\beta$, each time~$\beta$ is defined as the
  smallest power of two such that there is a size-$k$ CSP satisfying
  the current constraint~\cons whose longest undiscovered block is
  \betacrit. The set~$\Pi'$ contains exactly all integers obtained
  this way. Eventually,~$\beta<4$ and the algorithm branches by
  brute-force into all cases to set the breakpoints without violating
  the current constraint. Clearly, one of these cases is equivalent to
  a CSP satisfying this constraint. The algorithm verifies that this
  is indeed a CSP and that it has size at most~$k$ and correctly
  outputs the CSP. Hence, the algorithm is correct.

  It remains to show the running time of the algorithm. First, the
  for-each-loop in the main method is executed~$O(2^{\log n})=O(n)$
  times. Second,~by the restriction on~$\Pi'$, the
  repeat-loop in the main method is executed at most~$k$ times. To
  obtain the claimed running time, we bound the number of
  branches created in each call to~\splitf.
 
  In each call to \splitf the total length of the fragile pieces is
  less than $(2k)12(k^2+k)k\beta =24(k^4+k^3)\beta$: In the first
  call,~$\beta>n/2k$, so the bound holds. In the other cases, there
  are, by Observation~\ref{obs:piece-number} at most~$2k-2$ fragile
  pieces in~$x$ and~$y$. Furthermore, in this case~\splitf is called
  after \framesf. Thus, by~Lemma~\ref{lem:frame-size}, each fragile
  piece has length at most $12(k^2+k)k\beta$, and the overall bound
  follows.

  The procedure splits the fragile pieces into new pieces of length at
  most $\piecelen$ (i.e. there is a distance $\piecelen-1$ between the
  left endpoints of two consecutive pieces of the same splitting).
  Since~$\beta\ge 4$, we have $\piecelen-1\ge \beta/6$. Hence, this
  creates less than $144(k^4+k^3)$ new pieces of length $\piecelen$ plus
  at most one additional shorter piece at the end of each fragile
  piece. Hence, $145k^4$ is an upper bound on the number of new
  pieces. Branching for each piece into the case that it is solid or
  fragile can be done in $2^{145k^4}$ branches. The number of
  necessary branches for this part of \splitf can be reduced as
  follows: Since we merge series of consecutive pieces in $F$ or $S$,
  and since we do not need to consider branches with more than $k$
  solid pieces, we can directly look for the first and last piece of
  each~\betacrit block. This creates $O(\binom{145k^4}{4k})= O(
  \frac{(145k^4)^{4k}}{(4k)!})$ branches in each call of \splitf.

  The matching requires up to $k!$ branches, and the alignment at most
  $6^k$. Since $145^{4k}k!6^k=o((4k)!)$, we can bound the number of
  branches in each call of \splitf by $k^{16k}$. The \splitf~procedure
  is called at most $k$ times (by the restriction on~$\Pi$), thus
  creating  $O({\left(k^{16k}\right)}^k)=O(k^{16k^2})$ branches
  throughout the algorithm. Finally, the number of branches created in
  \framesf is~$\windowstime$~by~Lemma~\ref{lem:windows-rt}, and the
  number of branches created in the final brute-force can be bounded
  as follows. The length of the fragile pieces is~$O(k^4+k^3)$ and we
  need to guess at most~$2k-2$ precise breakpoint positions from this
  number. This can be done with~$k^{O(k)}$ branches. 
 
  Finally, note that all other steps of the algorithm can be clearly
  performed in polynomial time. Altogether, the total running time of
  the algorithm thus is
  \begin{align*}
O(k^{2k}n) \cdot k^{16k^2} \cdot k^{O(k)}  \cdot \windowstime \cdot \poly(n)\\
 =\overalltime.
  \end{align*}
\end{proof}

% 
% \begin{proof}[Sketched proof]
%  The main bottleneck in our algorithm is function \splitf, where the
% guessing of solid and fragile pieces requires to choose $4k$ pieces (the first
% and last of each series of solid pieces) among a total of $O(k^4)$
% pieces. This branching alone, repeated $k$ times, costs a factor $k^{16k^2}$.
%  The second bottleneck is the function \framesf, where up to $2k$
% guesses among $4k$ elements need to be performed, and this operation is repeated
% $2k$ times. Hence this function costs an additional $k^{4k^2}$. All other
% guesses create factors bounded by $o(k^{k^2})$, hence the total number of
% branches can be written $O(k^{21k^2})$.
% \end{proof}

\section{Putting Frames Next to Fixed Pieces}
%\subsection{Main Algorithm}
\label{sec:frames-aligned}
In this and the next section, we prove the two claimed lemmas
concerning \framesf. Informally, we show that, with the right constraint in the beginning, \framesf finds a constraint~$\cons$ that is satisfied by a size-$k$ CSP~$\sol$ whose longest undiscovered block has length at most~$2\beta-1$. Moreover,  the length of each fragile piece is~$O(k^3\beta)$ in every constraint produced by \framesf. The pseudo-code
of~\framesf is shown in~Algorithm~\ref{algo:frames-pc}.

The approach of~\framesf is to use a set of reduction rules to put
``frames'' into the fragile pieces, where a frame is an interval
within the fragile piece that contains \textit{all} breakpoints that
are contained in this piece. We call the actual shortest interval
containing all breakpoints of a fragile piece a ``window'', defined as
follows. Let~$\sol$ be a size-$k$ CSP satisfying~$\cons$, and let~$f$
be a fragile piece in~$\cons$. The \emph{window} of~$f$ is the
interval~$[a,b]$ such that~$\{a,\shiftr{a}1\}$ is the
leftmost breakpoint of~\sol in~$f$ and~$\{\shiftl{b}1,b\}$ is the
rightmost breakpoint of~$f$. Since a frame is required to contain all
breakpoints of a fragile piece it can be seen as a
``super''-approximation of the actual window. A formal definition of frames is as follows.
\begin{definition}
  Let~$\cons$ be a constraint.  A \emph{frame}~$[a,b]$ for a fragile
  piece~$f$ of~\cons is an interval that is contained in~$f$. A
  \emph{frame set} for~\cons~is a set~$\Phi$ of frames such that each
  fragile piece~$f$ contains at most one frame. A CSP~\sol~that
  satisfies~\cons{} \emph{satisfies} a frame set~$\Phi$ for~\cons{} if
  each breakpoint of~\sol is contained in a frame of~$\Phi$ or in a
  fragile piece without frame.
\end{definition}
Initially, the frame set is empty. Then the frames are added one by
one until each fragile piece has a frame. The approach to add the
frames to the constraint can be summarized as follows: first, we
compute an upper bound~$w$ on the length of the windows that only
depends on~$\beta$ and~$k$. Then, we apply a series of
\emph{frame~rules} that eventually place a frame in all fragile pieces
(Lines~4--5). As we will show, the frame length then depends on $w$
(and thus on~$k$ and~$\beta$) and on the maximum period length of the
unfixed (repetitive) solid pieces.  Since the frames contain all
breakpoints of \sol, it is possible to reduce fragile pieces until
they ``fit'' their frames (Line~6).  We now check whether there are
some unfixed solid pieces with a shortest period that is long compared to~$w$. If this
is the case, then the number of ``feasible'' alignments for these
pieces is small, and we can thus branch how to align these pieces
(Lines~7--11).  Formally, we call an alignment of $s=[a,b]$ and
$s'=[a',b']$ \emph{feasible} for \cons{} if the interval equidistant
to $[a,b]$ ($[a',b']$) from $s$ does not intersect any other solid
piece than $s'$ ($s$) in \cons. Note that each satisfying CSP can only
have feasible alignments, otherwise there is at least one fragile
piece without breakpoints.

Afterwards, we go back to applying the frame~rules (we will obtain shorter frames since the number of fixed pieces has
increased). If this is not the case, that is, all pieces have short periods 
compared to~$w$, then we show that the maximum frame length depends
only on $w$. Hence, in this case they are sufficiently short, and the
\framesf procedure has achieved its goal. The algorithm thus returns
to the main method where it calls \splitf to find new solid
pieces. 
\begin{algorithm*}[t]
  \caption{Procedure \framesf. Global variables: \cons, $\beta$.}
  \label{algo:frames-pc}
  \begin{noSpaceTabbing}
    \hspace*{0.7cm}\=\hspace*{0.5cm}\=\hspace*{0.5cm}\=\hspace*{0.6cm}\=\hspace*{4cm}\=\kill
  %  \framesf$(\cons,\beta)$\\    
    \lii{1} $w:=2\beta k + 1$ \> \> \> \> \com{upper bound on window length}\\
    \lii{2} \textbf{repeat} \textbf{:}\\
    \lii{3} \> Compute the maximum extension of each solid piece,\\ \>\> the piece graph $G[\cons,\Phi:=\emptyset]$, and the strips of each \reprep\ path  \\
    \lii{4} \> \textbf{while} there is a frameless fragile piece \textbf{:}\\
    \lii{5} \>\> place frames in fragile pieces by applying
Frame~Rules~1--6 \\
    \lii{6} \>\textbf{for each} fragile piece \textbf{:} % \\
    % \lii{7} \>\> 
    apply Fitting Rule~1\\
    %\lii{7} \>  $(\cons, $\textttup{new-align}$)\leftarrow $ \textttup{align-long-periods}(\cons)\\
		
		 \lii{7}\>  \textttup{new-align} $:=$ False\>\>\>\com{Fix pieces with long periods:}\\
    \lii{8}\>  \textbf{for each} repetitive solid piece~$s$ (with shortest period $\pi_s$) \textbf{:}\\
    \lii{9}\>  \>  \textbf{if} all fragile pieces adjacent to~$s$ or~$s'$ have length at most
$(12k^2+9k)\length{\pi_s}$ \textbf{:}\\
    \lii{10}\>  \> \> for each feasible alignment~branch into the case to add this alignment to~$R_S$\\		
    \lii{11}\>  \> \> \textttup{new-align} $\leftarrow$ True \\		
    \lii{12} \textbf{until} \textttup{new-align} = False\\
    \lii{13} \textbf{return} the modified constraint~$\cons$  
  \end{noSpaceTabbing}
%  \caption{Pseudo-code of \framesf.}
%  \label{algo:frames-pc}
\end{algorithm*}
%\subsection{Placing the Frames}
\label{sec:frame-rules}

In this section, we describe the frame~rules that place frames in
fragile pieces which are next to fixed pieces and some further simple
frame rules. Before doing so, we define two concepts that will be used
by the frame~rules: \emph{maximum extensions} and the \emph{piece
  graph}. Roughly speaking, maximum extensions are used locally to
bound the position of some breakpoints in the fragile pieces. The piece graph
provides a structural view of the relationship between pieces and
is used to show that one of the frame~rules can always be applied in
case there is a frameless fragile piece.
\paragraph*{Maximum extension of solid pieces.} Let~$s$ be a solid
piece in a constraint \cons. The \emph{maximum extension} of $s$ is
the interval $[\lext{s},\rext{s}]$ containing $s$ where~$\rext{s}$ and
$\lext{s}$ are defined as follows. If~$s$ is fixed, then let $\ell$ be
the largest integer such that 
$[s^*,\shiftr{s^*}\ell]\equiv[s'^*,\shiftr{s'^*}\ell]$,
and that no marker of $[s^*,\shiftr{s^*}\ell]$ or $[s'^*,\shiftr{s'^*}\ell]$ is in a solid piece other than $s$ or $s'$.
Then $\rext
s:=\shiftr{s^*}\ell$ and $\rext {s'}:=\shiftr{s'^*}\ell$.  
If $s$ is repetitive with a shortest period $\pi_s$, then let $a$ be the leftmost
marker in $s$ and define \rext s as the rightmost marker such that
the interval $[a,\rext s]$ has period~$\pi_s$, and that
no marker in $[a,\rext s]$ is in a solid piece other than $s$.
Marker \lext s is obtained symmetrically.  

The
following proposition is a straightforward consequence of the
definition of maximum extension.
\begin{proposition}\label{prop:aligned-in-extension}
  Let $s$ be a fixed solid piece, and let $[a,b]$ and $[c,d]$ be
  two intervals that are equidistant from~$s$ and such
  that~$[a,b]$ is contained in~$[\lext{s},\rext{s}]$ and~$[c,d]$ is
  contained in~$[\lext{s'},\rext{s'}]$. Then,~$[a,b]\equiv [c,d]$.
\end{proposition}
Note that, as a special case, the above proposition includes
single markers (that is, length-one intervals). The next
proposition simply states formally that the maximum extensions of a
solid piece contain the block which contains the solid piece.
\begin{proposition}\label{prop:block-limits}
  Let~$\cons$ be a constraint and~$s$ be a solid piece of~$\cons$. Any
  CSP that satisfies~$\cons$ has a block which contains~$s$ and is
  contained in~$[\lext s,\rext s]$.  Furthermore, let $f$ be a fragile
  piece next to $s$. Then, the window in $f$ contains at least one
  marker of~$[\lext s,\rext s]$.
\end{proposition}
\begin{proof}%[{of Proposition~\ref{prop:block-limits}}]
  If~$s$ is fixed, then the block containing~$s$ cannot contain the
  marker~$\shiftl{\lext s} 1$: if this marker is contained in the
  block, it is matched to~$\shiftl{\lext {s'}} 1$. By definition
  of~$\lext s$, either~$\shiftl{\lext s} 1 \not\equiv \shiftl{\lext
    {s'}} 1$ or one of $\shiftl{\lext s} 1$, $\shiftl{\lext {s'}} 1$
  belongs to a different solid piece~$t$. In the first case, we do not
  obtain a CSP; in the second case, there is at least one fragile
  piece without a breakpoint. Similarly, the block containing~$s$
  cannot contain~$\shiftr{\rext s} 1$.
  
  Every repetitive solid piece~$s$ is
  contained in a block that has a common shortest period~$\pi$ with $s$. By
  definition of~$\lext$ and~$\rext$ for repetitive solid pieces this
  block must thus be contained in~$[\lext{s},\rext{s}]$.
  
  Finally, consider a fragile piece $f$ to the right (to the left) of
  $s$. The window in $f$ contains the last (first) marker of the block
  containing $s$. By the above it thus contains at least one marker
  of~$[\lext{s},\rext{s}]$. 
\end{proof}

\begin{figure*}[t]
\centering
\includegraphics[scale=0.9]{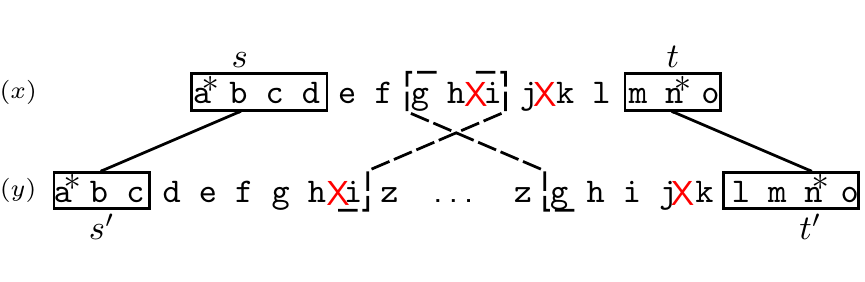}\hfill%
\includegraphics[scale=0.9]{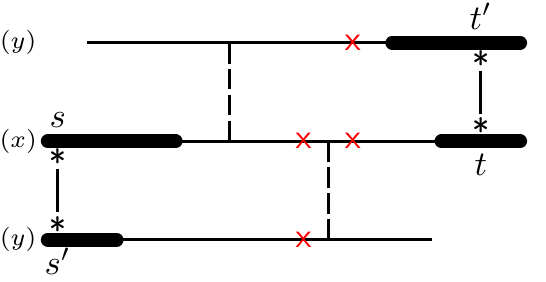}
 \caption{ Left: two pairs of fixed solid pieces $(s,s')$ and $(t,t')$.
Reference markers are shown with an asterisk, maximum extensions are
delimited with dashed lines, and the breakpoints of some possible CSP are
marked with red crosses.  
Right: a simplified representation of the same
pieces, where thick (resp. thin) lines are used for solid (resp. fragile)
pieces.}
\end{figure*}

\paragraph*{The~Piece~Graph.} 
Given a constraint \cons and a frame set~$\Phi$, the \emph{piece
  graph} $G[\cons,\Phi]$ is the bipartite graph $G:=(V_S\cup V_F, E)$
constructed as follows. 
\begin{itemize}
  \item $V_F$ contains one vertex $v_f$ for each frameless fragile piece $f\in F$,
  \item $V_S$ contains, for each repetitive solid piece $s\in S_x$ a
    vertex $v_s$, and for each fixed piece $s\in S_x$, two vertices
    $l_s$ and $r_s$ (for \emph{left} and \emph{right}). 
  \item For a fixed solid piece $s$ and a fragile piece $f\in F$,~$G$ contains the edge~$\{v_f,l_s\}$ if the last marker of $f$ is the first marker of~$s$ or of~$s'$, and the edge ~$\{v_f,r_s\}$ if the first marker of $f$ is the last marker of~$s$ or of~$s'$.
  \item For an unfixed solid piece $s$,~$G$ contains the
    edge~$\{v_f,v_s\}$ if the first marker of $f$ is the last marker
    of either~$s$ or~$s'$ or if the last marker of $f$ is the first
    marker of either~$s$ or~$s'$.
\end{itemize}
Note that the vertices $v_s$ or $l_s$ and $r_s$ are only defined for
pieces $s\in S_x$, but they represent both pieces $s$ and
$s'$. Observe furthermore that in case~$V_F\neq \emptyset$, there are
fragile pieces in~$\cons$ that do not have a frame
in~$\Phi$. Moreover, note that in this case the edge set of the piece
graph is nonempty. Our aim will thus be to gradually apply the
frame~rules until the piece graph is edge-less.  Each vertex is called
\emph{fragile}, \emph{fixed} or \emph{repetitive} depending on the
nature of the piece it represents. Note that most vertices of the
graph have degree at most 2, except for repetitive vertices which can
have degree up to 4. Vertices with smaller degree correspond initially
to the four pieces at the end of the sequences.

In order to deal seamlessly with pieces at the end of the input strings
(where no fragile piece is adjacent on one side), we introduce 
``phantom frames'' as follows. If $s$ contains the first
element of a string, i.e. $x[1]$ or $y[1]$, 
we say that $s$ has the \emph{phantom frame} $[x[0],x[1]]$ (resp. $[y[0],y[1]]$) 
to its left.
Likewise, if $s$ contains $x[n]$ or $y[n]$, it has the 
\emph{phantom frame} $[x[n],x[n+1]]$ (resp. $[y[n],y[n+1]]$) 
to its right. 
%laurent%
The idea behind phantom frames is as follows: solid pieces at the end of the strings
have one specific constraint, namely the first or last element is fixed. Since frames are 
used to bound the positions of endpoints of solid pieces, using phantom frames yields
that the constraint on end-of-string pieces is just a particular case of the general 
``endpoints in frames'' constraint.

We now have collected the prerequisites to state the frame~rules. A
frame rule is an algorithm that receives as input a constraint~\cons
and a frame set~$\Phi$ and updates both into a constraint~\cons' and a
frame set~$\Phi'$. A frame~rule is \emph{correct} if following holds. First, if
there is a size-$k$ CSP~$\sol$ satisfying~$\cons$ and~$\Phi$, then
there is also a size-$k$ CSP~$\sol'$ satisfying~$\cons'$
and~$\Phi'$. Second, the longest undiscovered block in~$\sol'$ is at
most as long as the longest undiscovered block in~$\sol$ (this
additional restriction will be used to argue that the choice
of~$\beta$ remains correct). Note that without loss of generality, we
describe all rules for pieces in~$x$ but they apply to fragile pieces
in~$x$ and~$y$. Furthermore, if a rule works on a single fixed vertex
in the piece graph, then we assume that this vertex is a left
vertex~$l_s$ (by inverting the instance one can also deal with all
right vertices). Finally, we state the additional frames of all rules
by defining an interval which contains the window, in order to ensure
that the frames are within the fragile pieces, we always intersect
this interval with the considered fragile piece~$f$.
\begin{figure*}[t]
\centering
\begin{tabular}{l@{\hspace{2.5cm}}l}
Frame Rule 1.&Frame Rule 2.\\
 \includegraphics[scale=1]{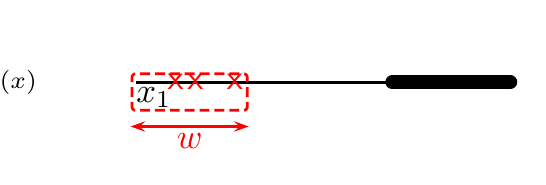}&\includegraphics[scale=1]{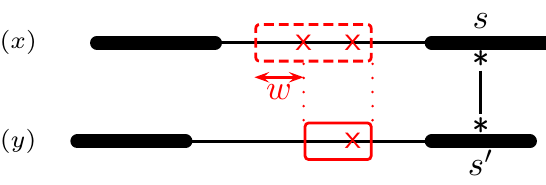}
 \\ 
%Frame Rule 3.&Frame Rule 4.\\
 %\includegraphics[scale=1]{rule3}&\includegraphics[scale=1]{rule4}
% \\ 
Frame Rule 3.&Frame Rule 4.\\
 \includegraphics[scale=1]{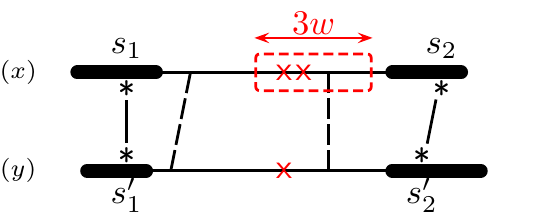}&\includegraphics[scale=1]{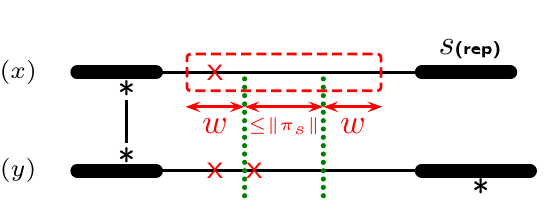}
 \\ 
Frame Rule 5.&Frame Rule 6.\\
 \includegraphics[scale=1]{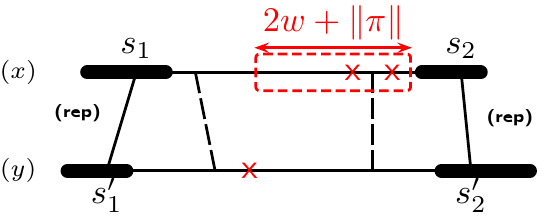}&\includegraphics[scale=1]{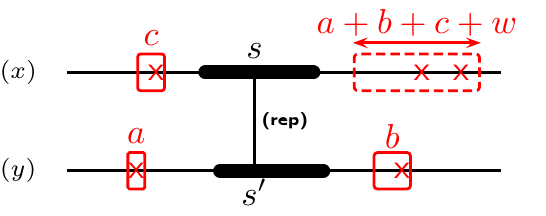}
\end{tabular}
 \smallskip
 \caption{ \label{fig:rules} Frame Rules 1--6  of 
\framesf. Frames are drawn as red boxes, the frame created at each step is dashed. %The strip in rule 6 is drawn as green dotted lines, and 
Possible breakpoint
positions  in \sol\ are shown as red crosses.}
\end{figure*}
The first rule puts frames into fragile pieces at the end of the string.
\begin{Frame Rule} 
\label{rule:fragile-end} If the piece graph
  contains a fragile degree-one vertex~$v_f$, then $f$ contains either
  $x[1]$ or $x[n]$.  If~$f$ contains~$x[1]$ add $f\cap [x[1], \shiftr{x[1]}{w}]$
  to~$\Phi$, otherwise add~$f\cap [\shiftl{x[n]}w, x[n]]$ to~$\Phi$.
\end{Frame Rule}
\begin{proof}[of the correctness of Frame~Rule~\ref{rule:fragile-end}]
  Fragile pieces of~$x$ that do not contain the first or the last
  marker of~$x$ are preceded and followed by a solid piece (since the
  splitting is alternating) and thus the corresponding vertex in the
  piece graph has degree two. Hence, a fragile piece in~$x$
  corresponding to a degree-one vertex of the piece graph contains
  either the first or the last marker of~$x$.  Assume without loss of
  generality that $f$ contains $x[1]$. The leftmost block of \sol\ in
  $x$ is necessarily an undiscovered block since it is contained in the
  fragile piece~$f$.
%  otherwise it would contain the first solid piece and there would
% be no breakpoint in $f$: a violation of Condition~\ref{prop:break-fragile}.
  Hence, marker $x[1]$ belongs to the first undiscovered block of \sol and
  it is next to a breakpoint of~\sol. Since the window (which contains
  all breakpoints in~$f$) has length at most~$w$, it is contained in the
  created frame $[x[1],\shiftr{x[1]}w]$. 
\end{proof}

\begin{Frame Rule} \label{rule:propagation} If the piece graph
  contains a degree-one vertex $l_s$ with neighbor~$v_f$ such that~$f$
  is next to~$s$ and~$s'$ does not contain~$y[1]$, then: 
	%let~$g$ be   the fragile piece to the left of~$s'$ in~$y$ and   let~$[\shiftl{s'^*}{u},\shiftl{s'^*}{v}]$ be the  frame of $g$; 
	let~$[\shiftl{s'^*}{u},\shiftl{s'^*}{v}]$ be the (possibly phantom)
	frame to the left of $s'$ in $y$;
	add
  the frame $f\cap [\shiftl{s^*}{(u+w-1))},\shiftl{s^*}{v}]$
  to~$\Phi$.
\end{Frame Rule}
\begin{proof}[of the correctness of Frame~Rule~\ref{rule:propagation}]
Consider first the case where $[\shiftl{s'^*}{u},\shiftl{s'^*}{v}]$
is a phantom frame: in this case, $\shiftl{s'^*}{v}$ is $y[1]$ and $u=v+1$. 
Hence, $y[1]$ is the first element of the block containing $s'$.
Since $y[1]$ and $\shiftl{s^*}{v}$ are equidistant from $s$, 
$\shiftl{s^*}{v}$ is the first element of the block containing $s$ 
 and the last element of the window in $f$. Since the window has length at
most $w-1$, it is contained in the frame $[\shiftl{s^*}{(v+w)},\shiftl{s^*}{v}]=[\shiftl{s^*}{(u+w-1))},\shiftl{s^*}{v}]$.

Consider now the (regular) case where $s'$ has a fragile piece $g$ to its left.  
By the frame definition, all breakpoints of a satisfying CSP \sol 
that are in~$g$ are within
   $[\shiftl{s'^*}{u},\shiftl{s'^*}{v}]$. Hence,  \shiftl{s'^*}{v} is in the same block as~$s'$. Consequently, the right limit of the window in~$f$ is to the left of  \shiftl{s^*}{v} in~$f$. Similarly, \shiftl{s'^*}{u} is in a different block than~$s'$ and thus there is a breakpoint to the right of \shiftl{s^*}{u} in~$f$. All other breakpoints in~$f$ can have distance at most~$w$ from this breakpoint. Hence, all breakpoints in~$f$ are contained in the created frame
  $[\shiftl{s^*}{(u+w-1)},\shiftl{s^*}{v}]$. 
\end{proof}

The above rules are relatively straightforward inferences of frame
positions that can be made because the piece graph has degree-one
vertices.  We now show some more intricate rules
that deal with the remaining cases. In particular, we show how to deal
with cycles in the piece graph.  We first consider cycles without
repetitive solid pieces. Note that the following rule performs a
branching. We thus extend the correctness notion to hold if there is
at least one branch in which the created constraint and
frame set can be satisfied.

\begin{Frame Rule} \label{rule:aligned-cycles} If the piece graph
  contains a simple cycle without repetitive vertices, then create one branch 
  for each edge~$\{v_f,u_s\}$ of this cycle. In each branch, add to~$\Phi$ the frame 
  \begin{itemize}
  \item 
    $f\cap [\shiftl{\lext{s}}{w}, \shiftr{\lext{s}}{(2w)}]$ if~$u_s=l_s$ for some solid piece~$s$, or
  \item $f\cap [\shiftl{\rext{s}}{2w}, \shiftr{\rext{s}}{(w)}]$ to~$f$ if~$u_s=r_s$ for some solid piece~$s$.
  \end{itemize}
\end{Frame Rule}

The following is a straightforward property of constraints and
satisfying solutions and used for showing the correctness of Frame~Rule~\ref{rule:aligned-cycles}.
\begin{proposition}\label{prop:keep-aligned}
 Let $s$ be a fixed solid piece in a constraint \cons. If markers $a$ and
$a'$ are equidistant from~$s$, then for any integer $i$, $\shiftr ai$ and $\shiftr
{a'}i$ are equidistant from~$s$.
 Moreover, given a CSP \sol satisfying \cons, the first markers (the last markers)
of the blocks of \sol containing $s$ and $s'$ are equidistant from~$s$.
\end{proposition}
\begin{proof}%[{of Property~\ref{prop:keep-aligned}}]
 The first part is directly obtained by definition: 
  %\begin{align*}
  \begin{equation*}   
  \size{s^*(\shiftr ai)}=\size{s^*a}+i
     =\size{s'^*a'}+i
     =\size{s'^*(\shiftr {a'}i)}.
  \end{equation*}
  For the second part, simply note that if~$s^*$ is at position~$j$ in
  the block containing~$s$, then~$s'^*$ is also at position~$j$
  in~$s'$. Hence, the first markers (and thus also the last markers)
  of both blocks are equidistant from~$s$.  
%  \end{align*}
\end{proof}

\begin{proof}[of the correctness of Frame Rule~\ref{rule:aligned-cycles}]
  Let $\mathfrak{P}$ be the set of CSPs that satisfy the constraint \cons\ and frame
  set~$\Phi$ and additionally have a minimum total length of undiscovered
  blocks. We show that
  there is a~$\sol\in \mathfrak{P}$  which has all breakpoints
  in~$[\shiftl{\lext{s}}{w}, \shiftr{\lext{s}}{(2w)}]$ for some
  vertex~$l_s$ of the cycle, thus showing correctness of the rule.

  Since the piece graph~$G[\cons,\Phi]$ is bipartite with partition~$V_S$
  and~$V_F$, the cycle alternates between vertices of $V_S$ and $V_F$.
  Moreover, all cycle vertices from $V_S$ are fixed, and alternate
  between left and right vertices (each fragile vertex of the cycle 
	is adjacent to a left vertex and to a right
  vertex). Hence there exist solid pieces $s_1,s_2,\ldots, s_\ell$ and
  fragile pieces $f_1,f_2,\ldots, f_\ell$ such that the cycle is
  $(l_{s_1}, v_{f_1}, r_{s_2}, v_{f_2}, \ldots, l_{s_{\ell-1}},
  v_{f_{\ell-1}}, r_{s_\ell}, v_{f_\ell})$. For simplicity, we consider
  indices only modulo $\ell$ (that is, $s_{\ell+1}=s_1$, $f_0=f_\ell$,
  etc.), and we assume that fragile pieces with odd indices are in $x$
  and those with even indices are in $y$.  
  % Figure~\ref{fig:proof-rule5a} shows an example proof.
\begin{figure}[t]\centering
\includegraphics{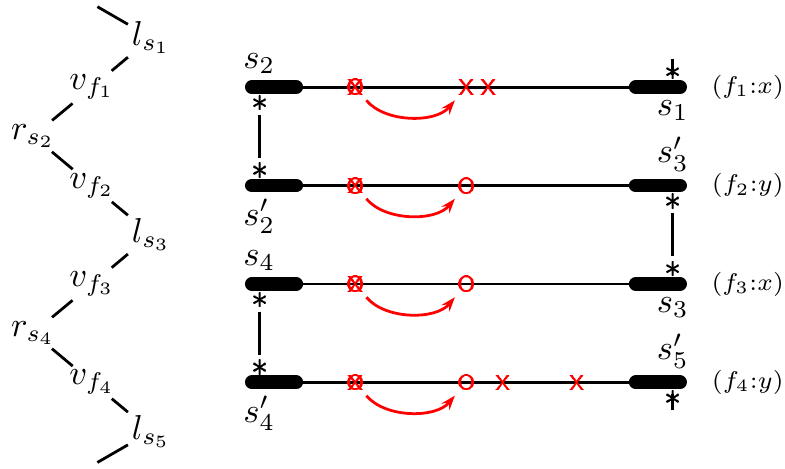}
\caption{\label{fig:proof-rule5a}Illustration for the first part of the correctness proof of Frame~Rule~\ref{rule:aligned-cycles}. If two fragile windows $f_i,f_j$ with different parity have several breakpoints (here, 
$i=1$ and $j=4$), then we can shift the position of the leftmost breakpoint in each fragile piece of the path to reduce the length of undiscovered blocks. The modifications (breakpoints added or deleted) are shown as red circles.}
\end{figure}
Consider a CSP~$\sol\in \mathfrak{P}$ such that there is no~$l_s$
whose window is contained in~$[\shiftl{\lext{s}}{w},
\shiftr{\lext{s}}{(2w)}]$. We transform this CSP into one that
fulfills this property. We first prove that in~$\sol$ either all
fragile pieces with odd or all fragile pieces with even indices
contain only one breakpoint. Assume towards a contradiction, that
there exist integers $i< j$ of different parity such that $f_{i}$ and
$f_{j}$ both have windows with at least two breakpoints and for each
$h$ with $i<h<j$, $f_h$ contains only one breakpoint.  Assume without
loss of generality that $i$ is odd and~$j$ is even. Hence,~$f_i$ is in
$x$ to the right of $s_{i+1}$ and $f_j$ is in $y$ to the right of
$s_j$.

For all $h$, $i\leq h\leq j$, let $a_h$ be the leftmost marker of the
window in $f_h$, and $b_h=\shiftr{a_h}1$.  For odd $h$, $a_h$ and
$a_{h+1}$ are the rightmost markers of the blocks containing $s_{h+1}$
and $s'_{h+1}$ and thus equidistant from $s_{h+1}$.  For even $h<j$, $b_h$
and $b_{h+1}$ are the left endpoints of the blocks containing
$s_{h+1}$ and $s'_{h+1}$, so they are equidistant from $s_{h+1}$. By
Proposition~\ref{prop:keep-aligned}, for all $i\leq h < j$, $[a_h,b_h]$
and $[a_{h+1},b_{h+1}]$ are equidistant from $s_{h+1}$. By definition of
$a_h$, the window in each $f_h$ is contained in $[a_h,\shiftr{a_h}w]$.
If one of these intervals is not contained in the maximum extension of
an adjacent solid piece, say~$[a_h,\shiftr{a_h}w]$ is not contained in
the maximum extension of~$s_{h+1}$, then~$\lext{s_{h+1}}$ is contained in~$[a_h,\shiftr{a_h}w]$. Hence, the window is contained in
$[\shiftl{\lext{s_{h+1}}}w,\shiftr{\lext{s_{h+1}}}w]$, contradicting our assumption on~$\sol$. In the following,
we thus assume that all intervals $[a_h,\shiftr{a_h}w]$ are contained in
the maximum extension of adjacent solid pieces, which by
Proposition~\ref{prop:aligned-in-extension} implies that they all have
the same content. In particular, this implies
$[a_i,\shiftr{a_i}w]\equiv [a_j,\shiftr{a_j}w]$.

We now describe a modification of~$\sol$ that results in a new CSP
which is not larger than~$\sol$, also satisfies the constraint and
frame set but has smaller total length of undiscovered blocks; the
modification is illustrated in
Figure~\ref{fig:proof-rule5a}. Let~$u+1$ and~$v+1$ be the lengths of
the leftmost undiscovered blocks in $f_i$ and $f_j$ respectively (assume
without loss of generality that $u\leq v$). These two undiscovered blocks are
thus $[b_i,\shiftr{b_i}u]$ and $[b_j,\shiftr{b_j}v]$, and they are
matched in \sol\ to other undiscovered blocks $[b_i',\shiftr{b_i'}u]$ and
$[b_j',\shiftr{b_j'}v]$. Note that since~$f_i$ is odd and~$f_j$ is
even,~$[b_i,\shiftr{b_i}u]$ is in a different string than
$[b_j,\shiftr{b_j}v]$. To create the new solution \sol' from \sol
apply the following modifications.  First, cut out $u+1$ markers from
the left of $[b_j',\shiftr{b_j'}v]$ (recall that~$u\le v$) which gives
two new blocks~$[b_j',\shiftr{b_j'}u]$
and~$[\shiftr{b_j'}(u+1),\shiftr{b_j'}v]$ if~$u<v$ and leaves the
block~$[b_j',\shiftr{b_j'}v]$ unmodified if~$u=v$. Now, match block
$[b_i',\shiftr{b_i'}u]$ to $[b_j',\shiftr{b_j'}u]$ (recall that these
blocks are in different strings). Now, shift the breakpoints of the
fragile pieces of the cycle as follows. For every odd $h$, $i<h<j$,
cut out $u+1$ markers from the left of $s_h'$ and $s_h$.  And for
every even $h$, $i<h\leq j$, add $u+1$ markers to the right of the
blocks containing $s_h$ and $s_h'$. Finally, in case~$u<v$, match the
shortened block~$\shiftr{b_j}(u+1),\shiftr{b_j}v$ to the
block~$\shiftr{b_j'}(u+1),\shiftr{b_j'}v$ created in the first
step. Note that by the discussion above, the pieces added to~$s_h$
and~$s'_h$ for even~$h$ have the same content. Hence, all matched
blocks have equal content. Furthermore, since the
block~$[b_i,\shiftr{b_i}u]$ is now unmatched, its markers are free to
be added to~$s_{i+1}$.

This new solution has at most as many blocks as $\sol$: we have
created at most one new breakpoint in~$[b_j',\shiftr{b_j'}v]$ and
removed a breakpoint in~$f_i$ by adding exactly~$u+1$ markers to the
right of~$s_{i+1}$. For all other fragile pieces~$f_h$, the breakpoint
has ``only'' been shifted to the right.  Furthermore,~$\sol'$
satisfies the same constraint \cons\ as \sol: the matching only
changed between undiscovered blocks which are not constrained. Moreover, the
fragile pieces for which the breakpoints have been modified are either
frameless (if they are on the cycle) or the modification adds a breakpoint that is between two breakpoints (in the modification of~$[b'_j,\shiftr{b'_j}v]$) However,the total length of the
undiscovered blocks has been reduced by~$2(u+1)$, which contradicts the
choice of \sol. 
We now know that in~$\sol$ the undiscovered blocks of
the cycle are either all in $x$ or all in~$y$. In the following, we
assume they are all in $x$, that is, in fragile pieces $f_j$ with odd
$j$. We now consider the following two cases: either there is no undiscovered
block, even in $x$, or there is at least one.

First consider the case that there is no undiscovered block in the cycle, that is, all the
windows contain only one breakpoint $[a_h,b_h]$.  If all markers
$b_h$ are within the maximum extensions of both adjacent solid pieces, we
create a new solution $\sol'$ from $\sol$ as follows: for every odd
$h$, cut out $b_h$ and $b_{h-1}$ from the left end of the blocks
containing $s_{h}$ and $s'_{h}$, and for every even $h$, add $b_h$ and
$b_{h-1}$ to the right end of the blocks containing $s_{h}'$ and
$s_{h}$. The solution $\sol'$ 
satisfies the same constraints as \sol, with the same total length of
undiscovered blocks. Repeat this operation of shifting the breakpoints to the
right until for some $i$ (without loss of generality, assume $i$ is even), $b_i$ is to the
right of $\rext{s_i}$. Then, the rule is correct, since for some branch the edge
$(r_{s_i}, v_{f_i})$ is selected and the frame $[\shiftl{\rext{s_i}}2w,
\shiftr{\rext{s_i}}w]$ which contains the only breakpoint of \sol in $f_i$ 
is added to $\Phi$.

\begin{figure}[t]\centering
\includegraphics{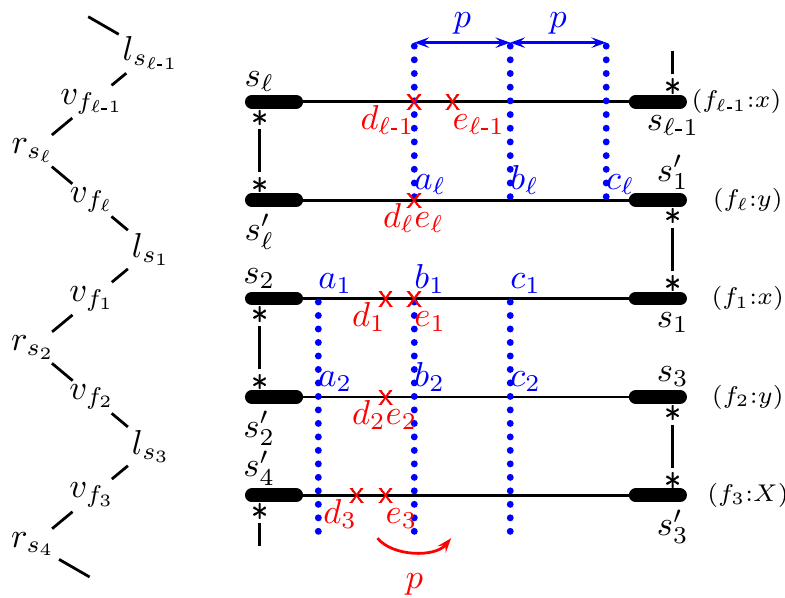}
\caption{\label{fig:proof-rule5b}Illustration for the correctness 
proof of the second part of~Frame~Rule~\ref{rule:aligned-cycles}. 
Given a cycle with total length of undiscovered blocks $p\geq 0$ we 
construct intervals $[a_h,b_h]$ and $[b_h,c_h]$ as shown (delimited 
by blue dotted lines). All the breakpoints in  intervals 
$[a_h,b_h]$ can be shifted to the corresponding $[b_h,c_h]$.}
\end{figure}
It remains to show the case where there is at least one undiscovered block in the fragile pieces of the cycle, that is, the total length $p$ of the
undiscovered blocks of the cycle is at least one. Note that by the choice
of~$w$,~$p<w$. We now show that the strings around the windows are
periodic with period length $p$, so that we can again shift all the
breakpoints of the fragile pieces to the right by steps of length $p$,
until at least one of them has distance at most $p$ from the end of a
maximum extension.
% Clearly $p<w$ is deduced from the fact that there are less than $k$
% short blocks concerned, and that they all have length strictly less
% than $2\beta $, thus $p<2k\beta =w$.

We first introduce some notations (see~Figure~\ref{fig:proof-rule5b}
for an illustration): for each $h$, let $[d_h,e_h]$ denote the window
of $f_h$. Let $b_1=e_1$, $a_1=\shiftl{b_1}p$,
$c_1=\shiftr{b_1}p$, and for each $h$, $2\leq h\leq \ell$, let~$a_h$,
$b_h$, and $c_h$ be the markers equidistant with $a_{h-1}$, $b_{h-1}$,
and $c_{h-1}$ from~$s_h$.

We first show that for every $h$ with $2\leq h\leq \ell$, we have
\begin{equation}\label{eq:ehbh_grows}
 \size{e_{h}b_h}=\size{d_{h-1}b_{h-1}}-1.
\end{equation}
For even values of $h$, $d_{h-1}$ and $d_h$ 
are equidistant from~$s_h$, so $\size{d_h b_h}=\size{d_{h-1} b_{h-1}}$. Since $f_h$ is in string $y$, it contains only one breakpoint, and thus
$\size{d_{h}e_{h}}=1$ and 
Equation~\eqref{eq:ehbh_grows} follows. 
For odd values of $h$, we have $\size{d_{h-1}e_{h-1}}=1$, and
$e_{h-1}$ and $e_h$ are equidistant from~$s_h$, thus
$\size{e_{h}b_h}=\size{e_{h-1}b_{h-1}}$, which also implies
equation~\eqref{eq:ehbh_grows}. 
% We first show that for every $h$ with $2\leq h\leq \ell$, we have
% \begin{equation}\label{eq:ehbh_grows}
%  \size{e_{h}b_h}=\size{e_{h-1}b_{h-1}}+\size{d_{h-1}e_{h-1}}-1
% \end{equation}
% For even values of $h$, $d_{h-1}$ and $d_h$ 
% are aligned wrt.~$s_h$, so $\size{d_h b_h}=\size{d_{h-1} b_{h-1}}$. Since $f_h$ is in string $y$, it contains only one breakpoint, and thus
% $\size{d_{h}e_{h}}=1$. Finally, $\size{d_h b_h} = \size{d_h e_h}+\size{e_h
% b_h}=$ and
% $\size{d_{h-1} b_{h-1}}= \size{d_{h-1} e_{h-1}}+\size{e_{h-1} b_{h-1}}$, and 
% equation~\eqref{eq:ehbh_grows} follows. 
% For odd values of $h$, we have $\size{d_{h-1}e_{h-1}}=1$, and
% $e_{h-1}$ and $e_h$ are aligned wrt.~$s_h$, thus
% $\size{e_{h}b_h}=\size{e_{h-1}b_{h-1}}$, which implies
% equation~\eqref{eq:ehbh_grows}. 
Hence the distance between window endpoint~$e_h$ and the marker $b_h$ increases, compared to the distance of~$e_{h-1}$ and~$b_{h-1}$ by the length
of the undiscovered blocks contained in the window of~$f_{h-1}$.  
This has two implications: first, in $f_\ell$, we have $\size{e_\ell b_\ell}=p$ and thus~$e_\ell=a_\ell$ (by definition $a_1$ has distance~$p$ from~$b_1$, and
this distance is conserved through the cycle). Second, for every~$j$, the
undiscovered blocks in~$f_j$ are contained in~$[a_j,b_j]$, and the window is contained in
$[\shiftl{a_j}1,b_j]$.

First, consider the case where each interval $[a_h, c_h]$ is contained in
the maximal extensions of both adjacent blocks. 
Thus, with Proposition~\ref{prop:aligned-in-extension}, 
we have $[a_h,b_h]\equiv[a_1,b_1]$ and $[b_h,c_h]\equiv[b_1,c_1]$ for
all $h$. 
We can now ``close'' the cycle: since $e_\ell$ and $e_1$ are the left
endpoints of the blocks containing $s_1'$ and $s_1$, they are equidistant from
$s_1$. Moreover, $e_\ell=a_\ell$ and $e_1=b_1$, so $a_\ell$ and $b_1$ are
equidistant from~$s_1$,
 which implies that $[a_\ell,b_\ell]\equiv[b_1,c_1]$. This now implies that, for all $h$,
$[a_h,b_h]\equiv [b_h,c_h]$. We now create a solution $\sol'$ from $\sol$ as
follows: for odd values of~$h$, cut out the $p$ leftmost markers from each block containing $s_h$ or
$s_h'$. For even values of~$h$, add $p$ markers to the right of blocks containing $s_h$
or $s_h'$ for even values of $h$. Match every undiscovered block that was matched
to some $[u,v]$ in some $f_h$ to $[\shiftr up, \shiftr vp]$ instead.
The solution $\sol'$ is again a CSP satisfying the same constraints,
with the same total length of undiscovered blocks but with all the breakpoints in the
cycle shifted
to the right by $p$ positions. Repeat this operation until for some $h$ the
interval $[a_{h}, c_{h}]$ is no longer contained in the maximal extension of the block to its right. Then, $[a_{h}, c_{h}]$
  contains $\lext{s_{h}}$, and interval $[\shiftl{a_h}1,b_h]$ is contained in $[\shiftl{\lext{s_h}}{2w},\shiftr{\lext{s_h}}w]$.
As argued above, the rule is correct if such a~$\sol\in \mathfrak{P}$ exists. Note that the modifications made in the proof do not increase the length of any undiscovered block. Hence, the second requirement for correctness is also satisfied.
% it remains to consider the case that $[a_{h}, c_{h}]$ is not included in the maximal extension of~$h$.   The rule correctly guesses the edge $(r_{s_{h}}, v_{f_h})$ and inserts the above frame. This proves the correctness of the rule.
%   we
% create a frame  $[\shiftl{\rext{s_h}}{2w},\shiftr{\rext{s_h}}w]$ which contains,
% in turns,
% $[\shiftl{\rext{s_h}}{(2p+1)},\shiftr{\rext{s_h}}p]$, 
% $[\shiftl{a_{h}}1, b_{h}]$ and the window contained in $f_h$.
\end{proof}

The rules presented so far deal with fixed solid pieces. In fact, if all
solid pieces are fixed, then these rules suffice to obtain frames in
all fragile pieces. With the following three rules, we thus deal with
the presence of repetitive solid pieces. 

\section{Frame Rules for Repetitive Pieces}
\label{sec:frames-repetitive}
In the rules, we have to deal with cycles in the piece graph that
 contain some repetitive vertices. We introduce the following concepts in
order to analyze the structure of paths between repetitive vertices
that contain fixed solid vertices.
\begin{figure}[t]\centering
 \includegraphics{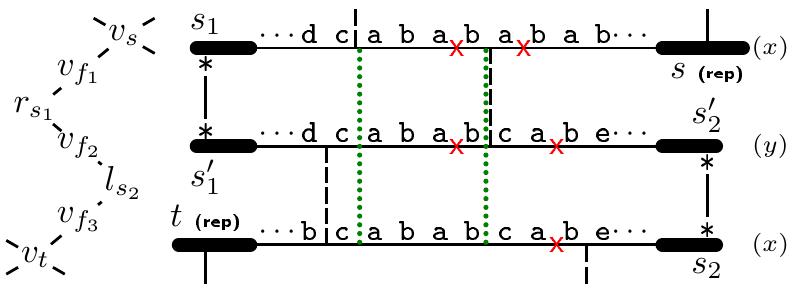}
 \caption{\label{fig:strips} Example of a \reprep\ path joining repetitive vertices
$v_s$ and $v_t$ (with respective periods \texttt{ab} and \texttt{ababc}), and
going through three fragile vertices and their adjacent fixed vertices.
 The strip of each fragile piece is delimited by the green dotted lines.}
\end{figure}
A $\reprep$ path~$(v_s, v_{f_1}, u_1,v_{f_2}, u_2, \ldots, u_{\ell-1}, v_{f_\ell} ,v_t)$ is a simple path of the piece graph such that the two endpoints~$v_s$ and~$v_t$ are repetitive vertices, and each~$u_i$ is a fixed solid vertex. Given a $\reprep$
path joining repetitive vertices $v_s, v_t$ and going through fragile
vertices $v_{f_1}, v_{f_2}, \ldots, v_{f_\ell}$, we define the
\emph{strip} of the path (see Figure~\ref{fig:strips}) as a set of intervals $\{I_{f_1}, I_{f_2},
\ldots, I_{f_\ell}\}$ such that:
\begin{enumerate}
\item Consecutive intervals~$I_{f_i},I_{f_{i+1}}$ are equidistant from
  the solid piece represented by~$u_i$. \label{prop:aligned-strips}
 \item Each interval $I_{f_i}$ is contained in the maximum extensions of 
 the two solid pieces next to $f_i$.
\label{prop:strips-in-extensions}
 \item The length of $I_{f_1}$ is maximal under Conditions~\ref{prop:aligned-strips} and~\ref{prop:strips-in-extensions}. \label{prop:maximal-strips}
\end{enumerate}
%The following simple properties of strips .....
\begin{proposition}\label{prop:nice_strips}
  All the strips in a \reprep\ path have the same length and
  content. Each interval of the strip is contained in its
  respective fragile piece.  Moreover, the strip of a \reprep~path is
  uniquely defined and computable in polynomial time.
\end{proposition}
\begin{proof}%[{of Proposition~\ref{prop:nice_strips}}]
  The fact that strips have the same length and content is a direct
  consequence of Proposition~\ref{prop:aligned-in-extension}, which can
  be applied according to Conditions~\ref{prop:aligned-strips}
  and~\ref{prop:strips-in-extensions}. Each strip must be contained in
  its fragile piece since it is in the intersection of the maximum
  extensions of the two adjacent solid pieces. 

  The second part of the claim can be seen by considering the
  following algorithm to compute the strip. First, check whether the
  strip is nonempty. That is, try the following for each marker~$a_1$
  in~$f_1$. Compute the marker~$a_2$ in~$f_2$ that is equidistant
  with~$a_1$ from~$u_1$. Then, compute the marker~$a_3$ in~$f_3$ that
  is equidistant with~$a_2$ from~$u_2$, and so on. If all~$a_i$'s are
  in the maximum extensions of both solid pieces next to~$f_i$, then
  the strip is nonempty. Otherwise, the length of~$I_{f_1}$ is
  zero. Now, assume the case that there was one~$a_1$ for which the
  above procedure is successful, that is,~$I_{f_1}$ contains one or
  more markers. Then, set~$I_{f_i}:=\{a_i\}$ for each~$i$. Now try to
  simultaneously expand all~$I_{f_i}$'s. That is, check whether one
  can add the marker to the left of each~$I_{f_i}$ without violating
  Condition~\ref{prop:strips-in-extensions} of the strip
  definition. If this is the case, then add these markers to
  the~$I_{f_i}$'s. If this is not the case, then continue by adding
  markers to the right until this is also not possible anymore. The resulting set of~$I_{f_i}$'s is the strip of the \reprep-path.
\end{proof}

\begin{proposition}\label{prop:breakpoints_close_to_strips}
  Let $\sol$ be any solution satisfying constraint \cons such that the
  total length of all windows in~$\sol$ is~$p$.  In each fragile
  piece $f$ of a \reprep\ path of \cons, writing $I_f=[c,d]$,
  the window of~$f$ is contained in~$[\shiftl cp, \shiftr dp]$.
	%the endpoints~$a$ and~$b$ of
  %the window~$[a,b]$ of~$f$ are either contained in~ or have distance at most $p$ from~$c$ or~$d$.
\end{proposition}
\begin{proof}%[{of Proposition~\ref{prop:breakpoints_close_to_strips}}]
We first introduce some notations: 
let $f_1, f_2, \ldots, f_\ell$ be the fragile pieces of the path, and, for every
$1\leq j\leq \ell$, let $[a_j,b_j]$ denote the window of~$f_j$,~$I_{f_j}=[c_j,d_j]$, $\alpha_j=\size{d_ja_j}$ and $\beta_j=\size{d_jb_j}$.
 
Hence we aim at showing that for all~$j$,
$\beta_j\leq p$, that is,~$b_j$ is either to the left or at at most~$p$ markers to the right 
of~$d_j$. The proof for the left bound, that is, to show that~$a_j$ is
at most~$p$ markers to the left of~$c_j$ is symmetrical.

By maximality of the strip length (Condition~\ref{prop:maximal-strips}),
the intervals of the strip cannot be extended to the
right. Condition~\ref{prop:strips-in-extensions} is the one
constraining the strip length, hence there exists a fragile piece
$f_{j_0}$ in the path such that this condition is tight, that is,
$d_{j_0}=\rext{s}$, where $s$ is the solid piece to the left of
$f_{j_0}$. Hence, $a_{j_0}$ is not to the right of~$d_{j_0}$, and thus
$\alpha_{j_0}=\size{d_{j_0}a_{j_0}}=\size{\rext{s}a_{j_0}}\leq 0$.

Now for all~$j$, $\beta_j-\alpha_j=\length{[a_j,b_j]}-1$, that is, it is
the length of the window contained in~$f_j$ minus one. Consequently,~$\beta_{j_0}<\length{[a_{j_0},b_{j_0}]}$. Moreover, for every $1\leq
j<\ell$, either the first markers of the window of $f_j$ and $f_{j+1}$
are matched and thus equidistant to the piece represented by~$u_i$ or
the last markers of the window of~$f_j$ and~$f_{j+1}$ are matched and
thus equidistant to~$u_i$. Hence,~either $\alpha_j=\alpha_{j+1}$ or
$\beta_j=\beta_{j+1}$. In the first case,~$\beta_{j+1}$ increases,
compared to~$\beta_j$, by at most~$\length{[a_{j+1},b_{j+1}]}-1$. Hence,~$\beta_j\le
\beta_{j_0}+p$ for all~$j\ge j_0$. By symmetry, the same holds for
all~$j\le j_0$. 
\end{proof}

The following rule serves as a ``preparation'' of our main rule that
deals with cycles containing repetitive vertices. It will ensure that
if there is a cycle containing repetitive vertices, then these
repetitive vertices will have have a common shortest period.
\begin{Frame Rule} \label{rule:repetitive-small-strips} If the piece
  graph contains a \reprep\ path between repetitive vertices~$v_s$
  and~$v_t$ with strip~$\{I_{f_1}, \ldots, I_{f_\ell}\}$ such that the
  %laurent% 
 strip $I_f=[u,v]$ in $f$ is shorter than the sum $\length{\pi_s}+\length{\pi_t}$ of
 a shortest periods of~$s$ and~$t$,
  then add the frame $f\cap [\shiftl uw, \shiftr vw]$ to $f$.
\end{Frame Rule}
\begin{proof}[of the correctness of Frame~Rule~\ref{rule:repetitive-small-strips}]
  By definition,~$w$ is at least the total length of the windows of \sol.  By
  Proposition~\ref{prop:breakpoints_close_to_strips}, the endpoints of the window of~$f$ thus have distance at
  most~$w$ from~$I_f$.  
\end{proof}

\begin{Frame Rule}
\label{rule:repetitive-cycle}
If Frame~Rule~\ref{rule:repetitive-small-strips} does not apply and the piece graph contains a simple cycle with repetitive vertices,
then do the following. Let $\length{\pi}$ be the length of a shortest period
of any repetitive solid piece in the cycle. Then, create one branch for each edge~$\{v_f,u_s\}$ of the cycle where~$u_s$ is a solid vertex for the solid piece~$s$. In each branch, add to~$\Phi$ the frame
\begin{itemize}
\item $f\cap [\shiftl{\rext{s}}{(\length{\pi}+w)},
  \shiftr{\rext{s}}{w}]$ if $f$ is to the right of $s$, or
\item $f\cap [\shiftl{\lext{s}}{w},
  \shiftr{\lext{s}}{(\length{\pi}+w)}]$  if $f$ is to the left of~$s$.
\end{itemize}
\end{Frame Rule}
\begin{proof}[of the correctness of Frame~Rule~\ref{rule:repetitive-cycle}]
  First, all repetitive pieces of the path have a common shortest
  period. Indeed, consider any two consecutive repetitive pieces $s$
  and $t$ of the cycle: they are linked by a \reprep\ path, in which
  we compute the strips. All strips in this path have equal length~$S$
  and also equal content (Proposition~\ref{prop:nice_strips}).  Hence,
  the maximal extensions of repetitive pieces~$s$ and $t$ have a
  common substring of length~$S$.  Since
  Frame~Rule~\ref{rule:repetitive-small-strips} does not apply, we
  have $S\geq \length{\pi_s}+\length{\pi_t}$. Thus, the maximum
  extensions of $s$ and $t$ contain a common substring longer than the
  sum of their respective shortest
  periods. By~Lemma~\ref{lem:periodicity}, each shortest period of~$s$ is a shortest period of~$t$ and vice versa. Thus all repetitive pieces of the cycle have a
  common shortest period~$\pi$.

Let $s_1, s_2, \ldots , s_\ell$ denote the repetitive pieces crossed successively by
the cycle (again, we write $s_{\ell+1}=s_1$). For each $i$, $1\leq i\leq \ell$, let
 $\lpiece xi$, $\rpiece xi$, $\lpiece yi$, $\rpiece yi$ be the fragile pieces to the left and right of
$s_i$ in $x$ and $s_i'$ in $y$, respectively. For each \reprep\ path of the cycle from
$s_i$ to~$s_{i+1}$, we say the path is \emph{positive} if the first vertex after
$s_i$ is~$v_{\lpiece xi}$ or~$v_{\rpiece yi}$, and \emph{negative} otherwise. In positive
\reprep\ paths, fragile pieces in $x$ are crossed from right to left (that
is, the solid piece to the right of the fragile piece is ``seen'' before the solid piece to its left),
and fragile pieces in $y$ are crossed from left to right. Thus a positive path
enters $s_{i+1}$ via either $v_{\rpiece x{i+1}}$ or $v_{\lpiece y{i+1}}$, and likewise a
negative path enters $s_{i+1}$ via either $v_{\lpiece x{i+1}}$ or $v_{\rpiece y{i+1}}$. 

First, consider the case that all windows are contained within the
strip and that both endpoints of the piece have distance at least $\length{\pi}$ 
to the borders of the strip.
% (by Proposition~\ref{prop:breakpoints_close_to_strips}, if this condition is false, 
% then the breakpoints outside the strip are at distance at most $w$ from the 
% border of the strip).
We show that in this case,
% unless some breakpoint is close to some maximum extension
% border,
we can shift all breakpoints in positive paths to the right by step
$\length{\pi}$ positions and all breakpoints in negative paths to the
 left by $\length{\pi}$ positions.
This is done as follows:
\begin{itemize}
\item For each fixed vertex $l_s$ in a positive path, cut out
  $\length{\pi}$ markers from the left of the blocks containing $s$
  and $s'$.
\item For each fixed vertex $r_s$ in a positive path, add $\length{\pi}$ markers from the right of the blocks containing $s$ and $s'$. 
\item For each fixed vertex $l_s$ in a negative path,
add $\length{\pi}$ markers to the left of the blocks containing $s$ and $s'$.
\item For each fixed vertex $r_s$ in a negative path, 
  cut out
$\length{\pi}$  markers from the right of the blocks containing $s$ and $s'$. 
 \item Replace each undiscovered block $[a,b]$ in a fragile
piece of a positive path by $[\shiftr {a}{\length{\pi}}, \shiftr
{b}{\length{\pi}}]$. % Let~$[a',b']$, be the block that~$[a,b]$ was matched to, then match $[a',b']$ to $[\shiftr {a}{\length{\pi}}, \shiftr
% {b}{\length{\pi}}]$.
 \item Replace each undiscovered block $[a,b]$, in a fragile piece of
a negative path by $[\shiftl {a}{\length{\pi}}, \shiftl
{b}{\length{\pi}}]$. % Update the matching as described above. match $[a,b]$ to  instead.
  \item For a repetitive vertex~$v_{s_i}$  such that the paths before and after~$v_{s_i}$ enter and leave~$v_{s_i}$ via the
same side (either  $\lpiece xi$ and  $\lpiece yi$, or $\rpiece xi$ and  $\rpiece yi$) either both paths are positive or both paths are negative. Apply the same operation as if the piece was fixed:
\begin{itemize}
\item If the path enters~$v_{s_i}$ via $\lpiece xi$ and leaves via $\lpiece yi$, then cut out the~$\length{\pi}$ leftmost markers of~$s$ and~$s'$ if the path is positive or add the~$\length{\pi}$ markers to the left of~$s$ and~$s'$ if the path is negative.
\item If the path enters~$v_{s_i}$ via $\rpiece xi$ and leaves via $\rpiece yi$, then cut out the~$\length{\pi}$ rightmost markers of~$s$ and~$s'$ if the path is negative or add the~$\length{\pi}$ markers to the right of~$s$ and~$s'$ if the path is positive.
\end{itemize}
\item For a repetitive vertex~$v_{s_i}$ such that the paths enter and leave the vertex via the
same string (either $\lpiece xi$ and  $\rpiece xi$, or $\lpiece yi$ and  $\rpiece yi$) it holds that the paths have the same orientation. Apply a similar operation as for a undiscovered block (assume without loss of generality that the path enters and leaves via~$x$):
\begin{itemize}
\item If~$v_{s_i}$ is between two positive paths then replace the block~$[a,b]$ of~$x$ containing~$s_i$ by $[\shiftr{a}{\length{\pi}},\shiftr{b}{\length{\pi}}]$.\item If~$v_{s_i}$ is between two negative paths then replace the block~$[a,b]$ of~$x$ containing~$s_i$ by $[\shiftl{a}{\length{\pi}},\shiftl{b}{\length{\pi}}]$.
\end{itemize}
\item For all other repetitive vertices, the paths enter from one string and leave via the other string and enter from one side and leave via the other side. Then the paths have opposite orientations; assume without loss of generality that the entering path is positive and the outgoing path is negative. Let~$[a,b]$ denote the block in~$x$ containing~$s_i$, and let~$[a',b']$ denote the block in~$y$ containing~$s_i'$. 
  \begin{itemize}
  \item If the cycle enters from $\lpiece yi$ and leaves via $\rpiece
    xi$, then replace~$[a,b]$ by~$[a,\shiftl{b}{\length{\pi}}]$ and~$[a',b']$ by~$[\shiftr{a'}{\length{\pi}},b']$ ($\length{\pi}$ markers are cut out of both blocks).
\item If the cycle enters from $\rpiece xi$ and leaves via $\lpiece
    yi$, then replace~$[a,b]$ by~$[a,\shiftr{b}{\length{\pi}}]$ and~$[a',b']$ by~$[\shiftl{a'}{\length{\pi}},b']$ (${\length{\pi}}$ markers are added to both blocks).
  \end{itemize}

\end{itemize}
Thus, all the breakpoints in fragile pieces have been shifted to the
right (in positive paths) or to the left (in negative paths) by a
period length ${\length{\pi}}$. Hence, this modification still gives a partition of both
strings. This partition has the same size as the original
one. Furthermore, it is also a common string partition which can be
seen as follows. The set of strings represented by the undiscovered blocks
of~$x$ and~$y$ remains exactly the same since they were shifted by the
period length. Hence, there is matching for the undiscovered blocks such that
each undiscovered block is matched to one representing the same string. For
the discovered blocks, the old matching remains a valid matching: The blocks
containing fixed solid pieces have both been modified on the same
side. Thus, they are either shortened by~${\length{\pi}}$ markers; in this case,
the matched blocks clearly represent equivalent strings. Or ${\length{\pi}}$
markers have been added on one side. In this case, the matched strings
are also equivalent, since the windows have distance at least~${\length{\pi}}$ to
the borders of the strip. The blocks containing repetitive pieces have
either been moved by~${\length{\pi}}$ positions, shortened by~${\length{\pi}}$ markers on
the same side,~${\length{\pi}}$ markers on the same side have been added, or they
have been shortened or extended on different sides. In the first three
cases, the strings represented by the new blocks remain equivalent for
the same reasons as for the blocks containing fixed solid pieces. It
remains to show the case in which blocks have been modified on
different sides. 

First, consider the case in which~$[a,b]$ is replaced
by~$[a,\shiftl{b}{\length{\pi}}]$ and~$[a',b']$ by~$[\shiftr{a'}{\length{\pi}},b']$. Since
the blocks are periodic with period length~${\length{\pi}}$ we
have~$[\shiftr{a'}{\length{\pi}},b']\equiv [{a'},\shiftl{b'}{\length{\pi}}]$. In the old
solution, this subinterval of~$[a',b']$ was matched
with~$[a,\shiftl{b}{\length{\pi}}]$, and thus~$[\shiftr{a'}{\length{\pi}},b']\equiv
[{a'},\shiftl{b'}{\length{\pi}}]\equiv [a,\shiftl{b}{\length{\pi}}]$.

Now consider the case in which~$[a,b]$ is replaced by~$[a,\shiftr{b}{\length{\pi}}]$ and~$[a',b']$
by~$[\shiftl{a'}{\length{\pi}},b']$. Since the blocks are periodic with period
length~${\length{\pi}}$ we have~$[a',b']\equiv
[\shiftl{a'}{\length{\pi}},\shiftl{b'}{\length{\pi}}]$. Since~$[a,b]\equiv [a',b']$ this
implies that the first~$\length{[a,b]}$ markers of the new blocks are
equivalent. Also because of the periodicity, we
have~$[b,\shiftr{b}{\length{\pi}}]\equiv
[\shiftl{b}{\length{\pi}},b]$. Since~$[\shiftl{b}{\length{\pi}},b]\equiv[\shiftl{b'}{\length{\pi}},b']$,
this implies that also the last~${\length{\pi}}$ markers of the new blocks are
equivalent.

Altogether, the modification gives a CSP of the same size, in which
the distance between the window endpoints and the strip endpoints has
decreased. The above operation can be repeated until at least one
breakpoint is at distance less than $\length{\pi}$ from the border of
a strip. In this case, all breakpoints of the corresponding path are
at distance at most $w+\length{\pi}$ from the border of their
corresponding strip (an argument similar to the proof of
Proposition~\ref{prop:breakpoints_close_to_strips} applies). In some
fragile piece $f$, the border of $I_f$ coincides with the maximum
extension of an adjacent solid piece $s$, thus, in $f$, the window is
contained in either $[\shiftl{\lext{s}}{(\length{\pi}+w)},
\shiftr{\lext{s}}{w}]$ or $[\shiftl{\rext{s}}{w},
\shiftr{\rext{s}}{(\length{\pi}+w)}]$. Since in one of the considered
branches, the rule adds the frame to this piece~$s$ and to the correct
side of the strip interval it is correct. Note that the modifications
made in the proof do not increase the length of any undiscovered
block. Hence, the second requirement for correctness is also
satisfied.
\end{proof}

The final case that needs to be considered is the one in which  the piece graph is acyclic but none of the other rules
applies. Then, the piece graph contains a repetitive degree-one vertex.
\begin{Frame Rule}\label{rule:repetitive-missing}
  %If the piece graph contains an edge~$\{v_s,v_f\}$ such that~$v_s$ is
  %repetitive and has degree one, then do the following. Let $a$,~$b$,
  %and~$c$ be the fragile pieces next to the repetitive pieces~$s$
  %and~$s'$ and assume without loss of generality that $a$ is to the
  %left of~$s'$ in~$y$, that~$b$ is to the right of $s'$ in~$y$, and
  %that $c$ is to the left of~$s$ in $x$. Let $[a_l, a_r]$, $[b_l,
  %b_r]$, and $[c_r, c_l]$ be the frames in $a$, $b$ and $c$.  (In case
  %where $s$ or~$s'$ contains a first or last marker of the sequence,
  %endpoint, for example if $s'$ starts in $y[1]$, let $a_r=y[1]$ and
  %$a_l$ be the ``phantom'' marker $\shiftl{y_1}1$).
	  If the piece graph contains an edge~$\{v_s,v_f\}$ such that~$v_s$ is
  repetitive and has degree one, then assume without loss of generality 
	that $f$ is to the right of $s$ in $x$, and do the following. 
	Let $[a_l, a_r]$, $[b_l, b_r]$, and $[c_l, c_l]$ be 
	the (possibly phantom) frames such that $[a_l, a_r]$ is to the
  left of~$s'$ in~$y$, that~$[b_l, b_r]$ is to the right of $s'$ in~$y$, and
  that $[c_l, c_r]$ is to the left of~$s$ in $x$.  
%$a_l=a_r=y_1$, $b_l=b_r=y_n$ or $c_l=c_r=x_1$
Add the frame $f\cap [f_l, f_r]$ to $f$, where
$f_l:=\shiftr{c_l}{(\size{a_rb_l}+1)}$ and
$f_r:=\shiftr{c_r}{(\size{a_lb_r}+w-2)}$.
\end{Frame Rule}
\begin{proof}[of the correctness of Frame~Rule~\ref{rule:repetitive-missing}]
  The window to the left and right of $s'$ in $y$ are contained in
  $[a_l,a_r]$ and $[b_l,b_r]$ respectively, and the window to the left
  of $s$ in $x$ is contained in $[c_l,c_r]$.  Consider the blocks
  containing $s$ and $s'$, and let $\ell$ be their length. The
  two endpoints of the block containing~$s'$ are
  in~$[\shiftr{a_l}1,a_r]$ and~$[b_l,\shiftl{b_r}1]$. Hence $\ell\geq
  \size{a_rb_l}$ and $\ell\leq
  \size{(\shiftr{a_l}1)(\shiftl{b_r}1)}=\size{a_lb_r}-2$.

  The leftmost marker of the block containing~$s$ is contained in
  $[\shiftr{c_l}1,c_r]$.  Thus, the rightmost marker (the one in $f$)
  is necessarily in $[\shiftr{c_l}{(\ell+1)},\shiftr{c_r}{(\ell)}]$
  which, by the above upper and lower bounds on~$\ell$, is contained
  in~$[\shiftr{c_l}{(\size{a_rb_l}+1)},\shiftr{c_r}{(\size{a_lb_r}-2)}]$.
  This marker is the leftmost marker of the window of~$f$ which has
  length at most $w$. Hence the frame
  $[\shiftr{c_l}{(\size{a_rb_l}+1)},\shiftr{c_r}{(\size{a_lb_r}+w-2)}]$
  contains the window of $f$. The rule is still correct if $s$ or $s'$
  corresponds to the end of a string, since the phantom frames 
  contain the leftmost or rightmost marker of the blocks
  containing~$s$ or~$s'$.
\end{proof}

After exhaustively applying the frame rules,  parts of fragile
pieces that are outside of frames do not contain a
breakpoint. Hence, we perform the following rule which shrinks fragile
pieces such that they fit their frame; at the same time, the solid
pieces are extended accordingly.
\begin{Fitting Rule}
  If there is a fragile piece~$f=[a,b]$ with frame~$[c,d]$ such
  that~$a\neq c$ or~$b\neq d$, then add~$[a,c]$ to the solid piece
  left of~$f$, add~$[d,b]$ to the solid piece right of~$f$, and
  set~$f:=[c,d]$.
\end{Fitting Rule}
\begin{figure}\centering
  
  \includegraphics[scale=1]{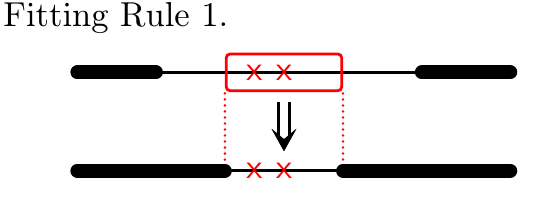}\hfill
  
  \caption{An illustration of Fitting Rule 1 of \framesf }
\label{fig:fitting-rule}
\end{figure}

We now show two important properties of instances for which none of
the frame rules applies. First, every fragile piece of these instances
has a frame. Second, the frame lengths 
are upper-bounded by a function of~$k$,~$\beta$, and the longest period of any
repetitive piece.
\begin{lemma}\label{lem:frame-size-with-period}
  Let~$\cons$ be a constraint with frame set~$\phi$ such that none of
  the Frame Rules~1--6 applies. Then, each fragile piece has a frame,
  and all frames have length at most~$(6k^2w+ 3kw+
  3k\max\{w,2\length{\pi}\})$, where, among the shortest periods of
  all repetitive solid pieces, $\pi$ denotes the longest one.
\end{lemma}
\begin{proof}%[of~Lemma~\ref{lem:frame-size-with-period}]
  First, we show that every fragile piece has a frame. If the piece
  graph contains a cycle, then either
  Frame~Rule~\ref{rule:aligned-cycles},~\ref{rule:repetitive-small-strips},
  or~\ref{rule:repetitive-cycle} applies. Otherwise, the piece graph is
  acyclic, and thus it either contains a degree-one vertex and one of
  the other Frame~Rules applies, or all vertices have degree zero
  which means that all fragile pieces have frames.

  Next, we show the upper bound on the frame length.  Let $L$ be the
  length of the longest frame created in this procedure, and let $\pi$
  be among the shortest periods of all repetitive pieces one with
  maximum length. We show that
\begin{equation}
  L\leq 6k^2w+ 3kw+ 3k\max\{w,2\length{\pi}\} \label{eq:frame-size}
\end{equation}
Let $h$ be the number of frames created before Frame~Rule~6 is first
applied, $1\leq h\leq 2k$.  Rules~1, 3, 4 and~5 produce frames of length
at most $(\max\{w,2\length{\pi}\}+2w)$. Since each application
of~Rule~2 increases the maximum frame length by $w$, all frames have
length at most $(\max\{w,2\length{\pi}\}+(h+1)w)$ before the first
application of~Frame~Rule~6.  Note that once Rule~6 is applied for the
first time, only Rules~2 and~6 can be applied.  We introduce the
following notations. A solid vertex (fixed or repetitive) is
\emph{closed} if all its adjacent fragile pieces have frames, and
\emph{open} otherwise. The \emph{weight} of an open vertex is the
total length of the frames in the adjacent fragile pieces. Let~$W$
denote the sum of the weights of all open vertices.

Before the first application of Frame Rule~6, $W\leq
3k[\max\{w,2\length{\pi}\}+(h+1)w]$ (for each solid piece $s$, the
weight of either $v_s$ or $l_s$ and $r_s$ together is at most the sum
of the weights of three different frames).  Afterwards, each time
Rule~2 or 6 is applied, an open vertex with some weight $u$ is
closed, and a frame of length $u+w$ is created in a fragile piece $f$
which is adjacent to at most one open vertex. Thus, the total weight
of open vertices $W$ is increased by at most $u+w-u=w$ with each
application of Frame~Rule~2 or~6. In the course of the algorithm, Frame~Rule~2 or~6 are applied at most $2k-h$
times. Hence $W$ never exceeds
 \begin{align*}
 &3k\left[\max\{w,2\length{\pi}\}+(h+1)w\right] +(2k-h)w \\
 &\leq 6k^2w+ 3kw+ 3k\max\{w,2\length{\pi}\}.
 \end{align*}
 Since no frame of length more than $W$ can be created, we have $L\leq
 W$, which proves the second part of the claim. 
\end{proof}
The bound given by the lemma above still contains the maximum~period
length~$\pi$ which means that it is too large to be useful for the
 \splitf procedure. However, the algorithm can now either find a
repetitive piece which can be fixed with few options 
(see Lemma~\ref{lem:long-period-few-options})
or the maximum period length is not too long.

\begin{restatable}{lemma}{lpoptionlem}\label{lem:long-period-few-options}
  Let~$\cons$ be a constraint that contains a repetitive solid
  piece~$s$ with shortest period~$\pi_s$ such that each fragile piece adjacent
  to~$s$ or~$s'$ has length at most~$(12k^2+9k)\length{\pi_s}$. Then,
  there are at most~$24k^2+18k$ feasible alignments,
  and any CSP  satisfying~$\cons$ matches elements of $s$ according to a feasible alignment.
%  there are at most~$12(k^2+k)$ different possibilities for a CSP to align~$s$
%  while satisfying~$\cons$ (they are called \emph{feasible} alignments).
\end{restatable}

\begin{proof}%[{of Lemma~\ref{lem:long-period-few-options}}]
The alignment corresponding
 to any CSP satisfying~\cons{} is necessarily feasible, since otherwise two
 distinct solid pieces would be contained in the same block.

  Without loss of generality, let~$\length{s}\ge \length{s'}$. Thus,
  in a satisfying CSP~$\sol$, either the leftmost marker of~$s$ is
  matched to a marker left of~$s'$ (or to the leftmost marker of $s'$),
	either
	the rightmost marker of~$s$ is
  matched to a marker right of~$s'$. 
	Consider the first case; by
  Condition~\ref{prop:break-fragile} of satisfying CSPs, the leftmost
  marker of~$s$ is matched to some marker in the fragile piece to the
  left of~$s'$. Note that since $s$ and $s'$ have a shortest period
  $\pi_s$, two different alignments are separated by a a multiple of
  $\length{\pi_s}$ markers.  Hence, there are at most~$12k^2+9k$
  different alignments in which the leftmost marker of~$s$ is matched
  to some marker of the fragile piece to the left of~$s'$. Similarly,
  there are at most~$12k^2+9k$ possible alignments in which the
  rightmost marker of~$s$ is matched to a marker of the fragile piece
  to the left of~$s'$ (for $\size{s}=\size{s'}$, it is possible that
  both left and right endpoints of $s$ are matched to markers of the
  fragile pieces to the left and right of $s'$). The total
  number of feasible alignments between~$s$ and~$s'$ thus is at
  most~$24k^2+18k$. 
\end{proof}
%\lpoptionlem*
% \begin{proposition}\label{propos:long-period-few-options}
%   Let~$\cons$ be a constraint that contains a repetitive solid
%   piece~$s$ with period~$\pi_s$ such that each fragile piece adjacent
%   to~$s$ or~$s'$ has length at most~$6(k^2+k)\length{\pi_s}$. 
%   There are at
%   most~$12(k^2+k)$ different possibilities to align~$s$ with an endpoint
% distance while
%   satisfying~$\cons$.
% \end{proposition}

By guessing the alignments of the long periods we have finally
achieved the goal of \framesf: all frames are ``short'' enough
to be split by \splitf.
\lemframesize*
%laurent% overall proof.
\begin{proof}%[of Lemma~\ref{lem:frame-size}]
  By~Lemma~\ref{lem:frame-size-with-period}, an instance in which no
  frame rule applies has frames of length at most~$(6k^2w+ 3kw+
  3k\max\{w,2\length{\pi}\})$ where~$\pi$ is a longest period among
  all repetitive pieces. In case~$2\length{\pi}\ge w$, this bound becomes
  $(12k^2+9k)\cdot \length{\pi}$. Consider any 
  repetitive piece~$s$ with period~$\pi$, the condition of Line~9 is satisfied:
  adjacent fragile pieces have length $\leq (12k^2+9k)\cdot \length{\pi}$. 
  Hence, at least one repetitive piece is fixed in the loop Lines 8--11, 
  and \textttup{new-align} is set
  ``True'', which means that the outer loop in \framesf will be
  repeated. Thus, when \framesf terminates, 
  we must be in the case~$2\length{\pi}\le w$ and  
  the frame size is bounded by 
  $6(k^2+k)w\le 12(k^2+k)k\beta$.
\end{proof}

The correctness of \framesf is simply a consequence of the
correctness of all single steps (always considering the correct branching in each branching step).
\windowslem*
\begin{proof}%[{of Lemma~\ref{lem:windows-correct}}]
  The correctness of all frame rules have already been proven.  The
  correctness of Fitting Rule~1 is trivial.  Finally, the correctness
  of Lines~8--11 follows simply from the fact that the alignment in
  one of the branches is the correct one (it considers all
  feasible alignments). Since the correctness definition of the frame
  rules demands that all undiscovered blocks are at most as long as before
  adding the frame, also the size bound for the longest undiscovered
  block holds.
\end{proof}

% \begin{lemma} \label{lem:windows-correct}
% If there exists an optimal CSP \sol\ satisfying \cons\ and such that all small
% blocks have length at most $2\beta $, 
% then there exists a guessing scenario such that $\framesf(\cons, \beta
% )$ 
% returns a new constraint $\cons'$ satisfied by
% $\sol$.
% \end{lemma}
%% if a piece has a long period (compared to adjacent fragile peices): it
% can be aligned
It thus remains to bound the running time of \framesf. In particular,
we need to show that the number of branches is bounded by a
function of~$k$.
% \begin{lemma} \label{lem:iterations}
% %\label{lem:iterations_windows}
% {\bf (a)}~The loop in  function \framesf terminates in $\leq 2k$
% iterations.
% In the main algorithm, 
%  {\bf (b)}~the
% inner loop is executed a total of $\leq 2k$ times.
% \end{lemma}
\windowsrt*
\begin{proof}%[{of Lemma~\ref{lem:windows-rt}}]
  First, note that the outer repeat-until~loop of \framesf is
  repeated at most~$2k$ times over the course of \emph{all} calls to
  \framesf: The procedure \framesf is called at most~$k$ times from
  the main method. Each additional time the repeat-until~loop is
  repeated, there is a pair of repetitive solid pieces that becomes a
  pair of fixed solid pieces at Line~10 of the previous pass of the
  repeat-until~loop. This can happen at most~$k$ times. 

  Second, note that the while~loop of Lines~4--5 is iterated at
  most~$2k$ times in each repetition of the other repeat-until~loop of
  \framesf: each rule creates exactly one frame, and, by
  Observation~\ref{obs:piece-number} there are at most $2k-2$ fragile
  pieces.
  
  Hence, there are at most~$4k^2$ times in which one of the frame
  rules at Line~5 creates branches and at most~$k$ times in
  which branches are created at Line~10. The only frame
  rules that perform branchings are
  Frame~Rules~\ref{rule:aligned-cycles}
  and~\ref{rule:repetitive-cycle}. In both cases, the rule branches
  into at most~$2k$ cases, since each cycle has at most~$k$ solid
  vertices and thus at most~$2k$ vertices edges in the cycle under
  consideration. Hence, the branchings performed by the frame~rules
  increase the running time by a factor of~$O((2k)^{4k^2})$. Each of the
  at most $k$ branchings in Line~10 is among at most~$24k^2+18k$ choices
  (Lemma~\ref{lem:long-period-few-options}).  Hence, these branchings
  increase the running time by a factor of~$O(k^{2k}\cdot
  k^k)$. Hence, the overall increase due to the branching is by a
  factor of~$\windowstime$; all other steps can be performed in
  polynomial time.
\end{proof}

% \begin{proof}[Sketched proof]
%  For {\bf (a)}, we show that
% the rules can create frames until all fragile pieces have a frame, hence there
% are as many iterations as fragile pieces. {\bf (b)} follows from
% Lemma~\ref{lem:split-correct}, where we prove that one new solid piece is
% discovered in each iteration of the outer loop. Finally, for {\bf (c)}, the
% total number of iterations is bounded by the number of repetitive pieces which
% become aligned ($\leq k$) plus the number of iterations of the outer loop ($\leq
% k$).
% \end{proof}

% \begin{proposition} \label{ lem:frame-size}
% When the algorithm terminates the inner loop (with \textttup{new-alignment} =
% False),
% every fragile piece has size at most $12(k^2+k)k\beta $. 
% \end{proposition}
% \begin{proof}[Sketched proof]
%   First of all, the
% frames created in function \framesf have a size which can be bounded by
% a function of~$k$,~$w$, and~$\length{\pi}$ the size of the longest period in any
% repetitive block. Then there are two cases: either $\length{\pi}\geq w$, in
% which case the period of the corresponding repetitive piece is large enough for
% the piece to be aligned in function \textttup{align-long-periods}, or
% $\length{\pi}\leq w$, in which case each frame has a size which can be bounded
% only by a function of $k$ and $w=2k\beta $.
% \end{proof}

\section{Conclusion}
We have presented the first fixed-parameter algorithm for MCSP
parameterized by the size of the partition. Aside from the
applications in comparative genomics, we believe that MCSP is a very
fundamental combinatorial string problem. Our work
thus makes a contribution to an area that has seen relatively few
advances on parameterized algorithms. An improvement of the very
impractical running time is desirable. Indeed, we believe that our
algorithm can be further improved to run in~$k^{O(k)}\cdot \poly(n)$
time.
However, a~$2^{O(k)}\cdot \poly(n)$ running
time is impossible with our approach of guessing the matching of solid
pieces and would need substantially new ideas.

Furthermore, it would be interesting to extend our result to
``signed'' MCSP~\cite{CZF+05,FLR+09,SMEM08} where each marker is
annotated with a direction and one may reverse blocks before matching.

\paragraph{Acknowledgments.}

We thank the anonymous reviewers of \textit{SODA} for valuable
feedback improving the presentation of this work.

\bibliographystyle{abbrvnat}
\bibliography{string-partition}

\begin{thebibliography}{14}
\providecommand{\natexlab}[1]{#1}
\providecommand{\url}[1]{\texttt{#1}}
\expandafter\ifx\csname urlstyle\endcsname\relax
  \providecommand{\doi}[1]{doi: #1}\else
  \providecommand{\doi}{doi: \begingroup \urlstyle{rm}\Url}\fi

\bibitem[Bafna and Pevzner(1996)]{BP96}
V.~Bafna and P.~A. Pevzner.
\newblock Genome rearrangements and sorting by reversals.
\newblock \emph{SIAM J. Comput.}, 25\penalty0 (2):\penalty0 272--289, 1996.

\bibitem[Bafna and Pevzner(1998)]{BP98}
V.~Bafna and P.~A. Pevzner.
\newblock Sorting by transpositions.
\newblock \emph{SIAM J. Discrete Math.}, 11\penalty0 (2):\penalty0 224--240,
  1998.

\bibitem[Bulteau et~al.(2013)Bulteau, Fertin, Komusiewicz, and Rusu]{BFKR13}
L.~Bulteau, G.~Fertin, C.~Komusiewicz, and I.~Rusu.
\newblock A fixed-parameter algorithm for minimum common string partition with
  few duplications.
\newblock In \emph{Proc.~13th~WABI}, volume 8126 of \emph{LNCS}, pages
  244--258, 2013.

\bibitem[Chen et~al.(2005)Chen, Zheng, Fu, Nan, Zhong, Lonardi, and
  Jiang]{CZF+05}
X.~Chen, J.~Zheng, Z.~Fu, P.~Nan, Y.~Zhong, S.~Lonardi, and T.~Jiang.
\newblock Assignment of orthologous genes via genome rearrangement.
\newblock \emph{IEEE/ACM T. Comput. Bi.}, 2\penalty0 (4):\penalty0 302--315,
  2005.

\bibitem[Cormode and Muthukrishnan(2007)]{CM07}
G.~Cormode and S.~Muthukrishnan.
\newblock The string edit distance matching problem with moves.
\newblock \emph{ACM Transactions on Algorithms (TALG)}, 3\penalty0
  (1):\penalty0 2, 2007.

\bibitem[Damaschke(2008)]{Dam08}
P.~Damaschke.
\newblock Minimum common string partition parameterized.
\newblock In \emph{Proc.~8th~WABI}, volume 5251 of \emph{LNCS}. Springer, 2008.

\bibitem[Fertin et~al.(2009)Fertin, Labarre, Rusu, Tannier, and
  Vialette]{FLR+09}
G.~Fertin, A.~Labarre, I.~Rusu, E.~Tannier, and S.~Vialette.
\newblock \emph{{Combinatorics of Genome Rearrangements}}.
\newblock Computational Molecular Biology. MIT Press, 2009.

\bibitem[Fine and Wilf(1965)]{FW65}
N.~J. Fine and H.~S. Wilf.
\newblock Uniqueness theorems for periodic functions.
\newblock \emph{Proceedings of the American Mathematical Society}, 16\penalty0
  (1):\penalty0 109--114, 1965.

\bibitem[Fu et~al.(2011)Fu, Jiang, Yang, and Zhu]{FJY+11}
B.~Fu, H.~Jiang, B.~Yang, and B.~Zhu.
\newblock Exponential and polynomial time algorithms for the minimum common
  string partition problem.
\newblock In \emph{Proc.~5th~COCOA}, volume 6831 of \emph{LNCS}, pages
  299--310. Springer, 2011.

\bibitem[Goldstein et~al.(2005)Goldstein, Kolman, and Zheng]{GKZ05}
A.~Goldstein, P.~Kolman, and J.~Zheng.
\newblock Minimum common string partition problem: Hardness and approximations.
\newblock \emph{Electron.~J.~Comb.}, 12, 2005.

\bibitem[Jiang et~al.(2012)Jiang, Zhu, Zhu, and Zhu]{JZZ+12}
H.~Jiang, B.~Zhu, D.~Zhu, and H.~Zhu.
\newblock Minimum common string partition revisited.
\newblock \emph{J.~Comb.~Optim.}, 23:\penalty0 519--527, 2012.

\bibitem[Kolman and Walen(2007)]{KW07}
P.~Kolman and T.~Walen.
\newblock Reversal distance for strings with duplicates: Linear time
  approximation using hitting set.
\newblock \emph{Electr. J. Comb.}, 14\penalty0 (1), 2007.

\bibitem[Shapira and Storer(2007)]{SS07}
D.~Shapira and J.~A. Storer.
\newblock Edit distance with move operations.
\newblock \emph{J. Discr. Alg.}, 5\penalty0 (2):\penalty0 380--392, 2007.

\bibitem[Swenson et~al.(2008)Swenson, Marron, Earnest-DeYoung, and
  Moret]{SMEM08}
K.~M. Swenson, M.~Marron, J.~V. Earnest-DeYoung, and B.~M.~E. Moret.
\newblock Approximating the true evolutionary distance between two genomes.
\newblock \emph{ACM J. Exp. Alg.}, 12, 2008.

\end{thebibliography}

\end{document}